\documentclass{aptpub}
\authornames{KRZYSZTOF \L ATUSZY\'{N}SKI AND DANIEL RUDOLF} 
\shorttitle{Convergence of hybrid slice sampling} 

\newcommand{\diam}{{\rm  diam}}

\newcommand{\vol}{{\rm  vol}}
\newcommand{\Borel}{\mathcal{B}}

\newcommand{\N}{\mathbb{N}}
\newcommand{\R}{\mathbb{R}}
\newcommand{\abs}[1]{\left\vert #1 \right\vert}	
\newcommand{\norm}[1]{\left\Vert #1 \right\Vert}	
\newcommand{\scalar}[2]{\langle #1 , #2 \rangle}
\newcommand{\set}[1]{\lbrace #1 \rbrace}

\renewcommand{\a}{\alpha}

\newcommand{\eps}{\varepsilon}

\renewcommand{\rho}{\varrho}
\newcommand{\dint}{\text{\rm d}}

\newcommand{\spec}{{\rm spec}}
\newcommand{\gap}{{\rm {g}ap}}
\newtheorem{alg}{Algorithm}
\newenvironment{algorithm}{\begin{alg}}{\end{alg}}

\begin{document}

\title{Convergence of hybrid slice sampling\\ via spectral gap\\} 

\authorone[University Warwick]{Krzysztof \L atuszy\'{n}ski}{}{}
\addressone{Department of Statistics, CV47AL Coventry, United Kingdom} 
\emailone{K.G.Latuszynski@warwick.ac.uk} 
\vspace*{-5ex}
\authortwo[Universit\"at Passau]{Daniel Rudolf}{}{}
\addresstwo{Universit\"at Passau, Innstra\ss e 33, 94032 Passau, Germany}
\emailtwo{daniel.rudolf@uni-passau.de}

\begin{abstract}	
	It is known that the simple slice sampler has robust convergence properties, however the class of problems where it can be implemented is limited. 
	In contrast, we consider hybrid slice samplers which are easily implementable and where another Markov chain approximately samples the uniform distribution on each slice. 
	Under appropriate assumptions on the Markov chain on the slice we show a lower bound and an upper bound of the spectral gap of the hybrid slice sampler in terms of the spectral gap of the simple slice sampler. 
	An immediate consequence of this is that spectral gap and geometric ergodicity of the hybrid slice sampler can be concluded from spectral gap and geometric ergodicity of its simple version which is very well understood.
	These results indicate that robustness properties of the simple slice sampler are inherited by (appropriately designed) easily implementable hybrid versions. 
	We apply the developed theory and analyse a number of specific algorithms such as the stepping-out shrinkage slice sampling, hit-and-run slice sampling on a class of multivariate targets and an easily implementable combination of both procedures on multidimensional bimodal densities.	
\end{abstract}



\section{Introduction}
Slice sampling algorithms are designed for 
Markov chain Monte Carlo
(MCMC) 
sampling from a
distribution given by a possibly unnormalised density. 
They belong
to the class of auxiliary variable algorithms that define a suitable Markov
chain on an extended state space. Following \cite{wang1987nonuniversal} and
\cite{edwards1988generalization} a number of different versions have been
discussed and proposed in
\cite{besag1993spatial,Hi98,mira2001perfect,MiTi02,MuAdMa10,Ne93,RoRo99,RoRo02,pmlr-v202-schar23a}.
We refer to these papers for details of algorithmic design and applications in
Bayesian inference and statistical physics. Here let us first focus on  the appealing \emph{simple slice
	sampler} setting in which no further algorithmic tuning or design by
the user is necessary: Assume that $K\subseteq \R^d$ and let the
unnormalised density be
$\rho\colon K \to (0,\infty)$. 
The goal is to sample approximately with respect to (w.r.t.) the distribution $\pi$ determined by $\rho$, i.e.
\[
\pi(A)=\frac{\int_A \rho(x)\,\dint x}{\int_K \rho(x)\,\dint x},\qquad A\in\Borel(K),
\]
where $\Borel(K)$ denotes the Borel $\sigma$-algebra.
Given the current state $X_n = x\in K$, the \emph{simple slice
	sampling} algorithm generates the next Markov chain instance $X_{n+1}$ by the following two steps:
\begin{enumerate}
	\item choose $t$ uniformly at random from $(0,\rho(x))$, 
	i.e. $t\sim \mathcal{U}(0,\rho(x))$;
	\item choose $X_{t+1}$ uniformly at random from 
	\[
	K(t):= \{x \in K \mid
	\rho(x) > t\},
	\]
	the level set of
	$\rho$ determined by $t$.
\end{enumerate}
The 
above defined \emph{simple slice sampler} transition mechanism is known to be reversible w.r.t. $\pi$ and
possesses very robust convergence properties that have been observed
empirically and established formally. 
For example Mira and Tierney \cite{MiTi02} proved that: If $\rho$ is bounded and the support of $\rho$
has finite Lebesgue measure, then the simple slice sampler is \emph{uniformly ergodic}.
Roberts and Rosenthal provide in \cite{RoRo99} criteria 
for \emph{geometric ergodicity}. 
Moreover, in \cite{RoRo99,RoRo02} the authors prove
explicit estimates of the total variation distance of the distribution of $X_n$ to $\pi$. In the recent work \cite{natarovskii2021quantitative}, depending on the volume of the level sets, an explicit lower bound of the \emph{spectral gap} of simple slice sampling is derived.

Unfortunately, the applicability of the simple slice sampler is
limited. In high dimensions sampling uniformly from the
level set of $\rho$ is in general infeasible and thus the second step
of the algorithm above can not be performed. Consequently, the second
step is replaced by sampling a Markov chain on the level set, which has the uniform distribution as the invariant one. Following the terminology of
\cite{RoRo97} we call such algorithms \emph{hybrid slice
	samplers}. 
We refer to  \cite{Ne03} where various procedures and designs for the
Markov chain on the slice are
suggested and insightful expert advice is given.

Although being easy to implement, \emph{hybrid slice sampling} in general has not been 
analyzed theoretically and little\footnote{A notable exception is elliptical slice sampling \cite{MuAdMa10}, that has been recently investigated in \cite{natarovskii2021geometric}, where a geometric ergodicity statement is provided.} is known about its convergence properties. The
present paper is aimed at closing this gap by providing statements about the inheritance of convergence from the simple to the hybrid setting.

To this end we study the 
absolute spectral gap of hybrid slice
samplers. The 
absolute
spectral gap of a Markov operator $P$ or a corresponding Markov 
chain $(X_n)_{n\in \N}$ is given by
\[
\gap(P)=1-\norm{P}_{L^0_{2,\pi} \to L^0_{2,\pi}},
\]
where $L^0_{2,\pi}$ is the space of functions $f\colon K \to \R$ with zero mean 
and finite variance (i.e. $\int_K f(x) \dint
\pi(x) = 0;$ $\norm{f}_{2}^2 = \int_K \abs{f(x)}^2 \dint
\pi(x)<\infty$) and 
$\norm{P}_{L^0_{2,\pi} \to L^0_{2,\pi}}$ denotes the
operator norm. We refer to \cite{Ru91} for the functional analytic
background. From the computational point of view, existence of the
spectral gap (i.e. $\gap(P)>0$) implies a number of desirable
and well studied robustness properties. In particular 
\begin{itemize}
	\item the spectral gap implies geometric ergodicity
	\cite{kontoyiannis2009geometric,RoRo97} and the variance bounding property \cite{RoRo08};
	\item for reversible
	Markov chains the spectral gap implies that a CLT holds for all functions $f \in
	L_{2,\pi}$, c.f.  \cite{Ge92,kipnis1986central}; 
	\item furthermore, consistent estimation of the CLT asymptotic
	variance is well established for geometrically ergodic chains
	(c.f. \cite{bednorz2007few, flegal2010batch, hobert2002applicability, JoHaCaNe06}).
\end{itemize}

Additionally, quantitative information on the spectral gap allows the formulation of 
precise non-asymptotic statements.
In particular, it is well known, see e.g. \cite[Lemma~2]{NoRu14},
that if $\nu$ is the initial distribution of the reversible Markov chain in
question, i.e. $\nu=\mathbb{P}_{X_1},$ then
\[
\norm{\nu P^n - \pi}_{\mbox{tv}} \leq (1-\gap(P))^n \norm{\frac{d\nu}{d \pi}-1}_{2},
\]
where $\nu P^n = \mathbb{P}_{X_{n+1}}$. See \cite[Section~6]{Ba05} for a
related $L_{2,\pi}$ convergence result. Moreover, when considering the
sample average, one obtains 
\[
\mathbb{E} \abs{ \frac{1}{n} \sum_{j=1}^n f(X_j) - \int_K f(x) \dint\pi(x)}^2 
\leq \frac{2}{n\cdot \gap(P)} 
+ \frac{c_p\norm{\frac{d\nu}{d\pi}-1}_{\infty}}{n^2 \cdot \gap(P)},
\]
for any $p>2$ and any 
function 
$f \colon K \to \R$ with
$\norm{f}_{p}^p 
= \int_K \abs{f(x)}^p  \pi(\dint x) \leq 1$, 
where $c_p$ is an explicit
constant
which depends only on $p$. One can also take a burn-in into account, 
for further details see \cite[Theorem~3.41]{Ru12}.
This indicates that the spectral gap of a Markov chain is
central to robustness and a crucial quantity in both asymptotic and non-asymptotic analysis of MCMC estimators.

The route we endeavour is to conclude the spectral gap of the hybrid
slice sampler from the more tractable spectral gap of the simple slice sampler. 
So what is known about the spectral gap of the simple slice sampler?
For saying more on this we require the following notation. Define $v_\rho \colon [0,\infty) \to [0,\infty]$ by $v_\rho(t):= \vol_d(K(t))$, which for level $t$ returns the volume of the level set. We say for $m\in\mathbb{N}$ that $v_\rho \in \Lambda_m$ if 
\begin{itemize}
	\item $v_\rho$ is continuously differentiable and $v_\rho'(t)<0$ for any $t\geq0$; and
	\item the mapping $t\mapsto t v_\rho'(t)/v_\rho(t)^{1-1/m}$ is decreasing on the support of $v_\rho$.
\end{itemize}
Recently, in \cite[Theorem~3.10]{natarovskii2021quantitative} it has been shown that, if $v_\rho \in \Lambda_m$, then $\gap(U)\geq 1/(m+1)$. This provides a criterion for the existence of a spectral gap as well as a quantitative lower bound, essentially depending on whether $t\mapsto t v_\rho'(t)/v_\rho(t)^{1-1/m}$ is decreasing or not.

Now we are in a position to explain our contributions.
Let $H$ be the Markov kernel
of the hybrid slice sampler determined by a family of transition
kernels $H_t$, where each $H_t$ is a Markov kernel with uniform limit distribution, say $U_t$,
on the level determined by $t$.
Consider 
\[
\beta_k := \sup_{x\in K} \left(
\int_0^{\rho(x)}\norm{H^k_t-U_t}_{L_{2,t}\to L_{2,t}}^2 \,\frac{\dint t}{\rho(x)}\right)^{1/2},
\]
and note that the quantity $\norm{H^k_t-U_t}_{L_{2,t}\to L_{2,t}}^2$ measures
how fast $H_t$ gets close to $U_t$. Thus
$\beta_k$ 
is the supremum over expectations of a function which measures the speed of convergence
of $H_t^k$ to $U_t$. The main result is stated 
in Theorem~\ref{thm: low_upp_spec} and it is as follows:
Assume that $\beta_k \to 0$ for increasing $k$ and assume $H_t$ 
induces a positive semi-definite Markov operator for every level $t$.
Then
\begin{equation}  \label{eq: ineq}
\frac{\gap(U)-\beta_k}{k} \leq \gap(H) \leq \gap(U), \quad k\in\N . 
\end{equation}
The first inequality implies that whenever there exists a spectral gap 
of the simple slice sampler and $\beta_k \to 0$, 
then there is a spectral of the hybrid slice sampler.
The second inequality of \eqref{eq: ineq} verifies a very intuitive
result, namely that the simple slice sampler is always better than the hybrid one.

We demonstrate how to apply our main theorem in different settings.
First, we consider a stepping-out shrinkage slice sampler, suggested in \cite{Ne03}, in 
a simple bimodal $1$-dimensional setting. Next we turn to the
$d$-dimensional case and on each slice perform a single step of the hit-and-run algorithm, studied in
\cite{BeRoSm93,Lo99,Sm84}. 
Using our main theorem we prove
equivalence of the spectral gap (and hence geometric ergodicity) of
this hybrid hit-and-run on the slice and the simple slice sampler.
Let us also mention here that in \cite{RuUl15} 
the hit-and-run algorithm, hybrid hit-and-run on the slice 
and simple slice sampler are compared, according to covariance ordering \cite{Mi01}, 
to a random walk Metropolis algorithm.
Finally, we combine the stepping-out shrinkage and hit-and-run slice
sampler. The resulting algorithm is practical and easily implementable
in multidimensional settings. For this version we again show
equivalence of the spectral gap and geometric ergodicity with the
simple slice sampler for multidimensional bimodal
targets. 

Further note that we consider single auxiliary variable 
methods to keep the arguments simple. We believe that
a similar analysis
can also be done if one considers multi auxiliary variable methods.

The structure of the paper is as follows. In Section~\ref{sec: basic}
the notation and preliminary results are provided. These include a necessary and sufficient condition
for reversibility of hybrid slice sampling in Lemma~\ref{lem: crit_rev} followed by a useful representation of slice samplers
in Section~\ref{subsec:useful},
which is crucial in the proof of the main result.
In Section~\ref{sec: spectral_gap} we state and prove the main result.
For example in Corollary~\ref{thm: upp_est_op_norm} a lower bound of the spectral
gap of a hybrid slice sampler is provided which performs several steps w.r.t.
$H_t$ on the chosen level.
In Section~\ref{sec: example}
we apply our result to analyse a number of specific hybrid slice
sampling algorithms in
different settings that include multidimensional bimodal distributions.

\section{Notation and basics} \label{sec: basic}

Recall that $\rho: K  \to (0, \infty)$ is an unnormalised density on
$K \subseteq \R^d$ and denote the level set of $\rho$ as
\[
K(t)=\set{ x\in K\mid \rho(x)> t }.
\]
Hence the sequence $(K(t))_{t\geq 0}$ of subsets of $\R^d$ satisfies
\begin{enumerate}
	\item $K(0) = K$;
	\item $K(s) \subseteq K(t)$ for $t<s$;
	\item $K(t) = \emptyset$ for $t\geq\norm{\rho}_{\infty}$.
\end{enumerate}
Let $\vol_d$ be the $d$-dimensional Lebesgue measure  and let
$(U_t)_{t\in (0,\Vert \rho \Vert_\infty)}$ be a sequence of distributions, where $U_t$ is the uniform distribution on $K(t)$, i.e.
\[
U_t(A)=\frac{\vol_d(A\cap K(t))}{\vol_d(K(t))}, \quad A\in \Borel(K).
\]
Further let $(H_t)_{t\in (0,\Vert \rho \Vert_\infty)}$ be a sequence of transition kernels, where  $H_t$ is a transition kernel on $K(t)\subseteq \R^d$.
For convenience we extend the definition of the transition kernel $H_{t}(\cdot,\cdot)$ on the measurable space $(K,\Borel(K))$.
We set 
\begin{align}  \label{al: path_case}
\bar{H}_t(x,A)& = \begin{cases}
0 &  x\not\in K(t),\\
H_t(x,A\cap K(t)) & x\in K(t). 
\end{cases}
\end{align}
In the following we write $H_t$ for $\bar{H}_t$ and consider $H_t$ as extension on $(K,\Borel(K))$.
The transition kernel of the hybrid slice sampler is given by
\[
H(x,A) = \frac{1}{\rho(x)}\int_0^{\rho(x)} H_t(x,A)\, \dint t, \quad x\in K,\, A\in \Borel(K).
\]
If $H_t=U_t$ we have the simple slice sampler studied in \cite{MiTi02,natarovskii2021quantitative,RoRo99,RoRo02}. 
The transition kernel of this important special case is given by
\[
U(x,A) = \frac{1}{\rho(x)}\int_0^{\rho(x)} U_t(A)\, \dint t, \quad x\in K,\, A\in \Borel(K).
\]
We provide a criterion for reversibility of $H$ w.r.t.
$\pi$. Therefore let us define the density 
\[
\ell(s) = \frac{\vol_d(K(s))}{\int_0^{\Vert \rho \Vert_\infty} \vol_d(K(r))\,\dint r},\quad s\in(0,\Vert \rho \Vert_\infty),
\]
of the distribution of the level sets on $((0,\Vert \rho \Vert_\infty),\Borel((0,\Vert \rho \Vert_\infty))).$

\begin{lemma}  \label{lem: crit_rev}
	The transition kernel $H$ is reversible w.r.t. $\pi$ iff 
	\begin{equation}	\label{eq: rev_equi}
	\int_0^{\Vert \rho \Vert_\infty} \int_B H_t(x,A)\, U_t(\dint x)\, \ell(t)\dint t
	= \int_0^{\Vert \rho \Vert_\infty} \int_A H_t(x,B)\, U_t(\dint x)\, \ell(t)\dint t, \quad A,B \in \Borel(K).
	\end{equation}
	In particular, if $H_t$ is reversible w.r.t. $U_t$ for almost all $t$ (concerning $\ell$), 
	then $H$ is reversible w.r.t. $\pi$.  
\end{lemma}
Equation \eqref{eq: rev_equi} is the detailed balance condition of $H_t$ w.r.t. $U_t$ in average sense,
i.e.
\[
\mathbb{E}_\ell [H_\cdot(x,\dint y) U_\cdot(\dint x)] = \mathbb{E}_\ell [H_\cdot(y,\dint x) U_\cdot(\dint y)],\quad x,y\in K.
\]
Now we prove Lemma~\ref{lem: crit_rev}.
\begin{proof}
	First, note that
	\begin{align*}
	\int_K \rho(x)\, \dint x & = \int_0^{\Vert \rho \Vert_\infty} \int_K \mathbf{1}_{(0,\rho(x))} (s)  \, \dint x\, \dint s \\
	& = \int_0^{\Vert \rho \Vert_\infty} \int_K \mathbf{1}_{K(s)} (x) \, \dint x\, \dint s 
	= \int_0^{\Vert \rho \Vert_\infty} \vol_d({K(s)}) \, \dint s.
	\end{align*}
	By this, we obtain for any $A,B\in \Borel(K)$ that
	\begin{align*}
	& \int_B H(x,A)\, \pi(\dint x) = \int_B \int_0^{\rho(x)} H_t(x,A) \frac{\dint t}{\int_0^{\Vert \rho \Vert_\infty} \vol_d(K(s))\dint s} \dint x\\
	& = \int_B \int_0^{\Vert \rho \Vert_\infty} 
	\mathbf{1}_{K(t)}(x)H_t(x,A) \frac{\ell(t)}{\vol_d(K(t))} \dint t\dint x
	=
	\int_0^{\Vert \rho \Vert_{\infty}} \int_B H_t(x,A)\, U_t(\dint x)\,\ell(t) \dint t.
	\end{align*}
	As an immediate consequence from the previous equation we have the claimed equivalence
	of reversibility and \eqref{eq: rev_equi}.
	By the definition of the reversibility of $H_t$ according to $U_t$ holds
	\[
	\int_B H_t(x,A)\, U_t(\dint x) = \int_A H_t(x,B)\, U_t(\dint x).
	\]
	This, combined with \eqref{eq: rev_equi}, leads to the reversibility of $H$.
\end{proof}

We always want to have that $H$ is reversible w.r.t. $\pi$.
Therefore we formulate the following assumption.
{\assumption
	Let $H_t$ be reversible w.r.t. $U_t$ for any $t\in(0,\Vert \rho \Vert_{\infty})$.
}

Now we define Hilbert spaces of square integrable functions and Markov operators.
Let $L_{2,\pi}=L_2(K,\pi)$ be the space of
functions $f\colon K \to \R$ which satisfy
$\norm{f}_{2,\pi}^2 := \langle f,f \rangle_{\pi} <\infty$, where 
\[
\langle f,g \rangle_\pi :=\int_K f(x)\, g(x)\, \pi(\dint x)
\]
denotes the corresponding inner-product of $f,g \in L_{2,\pi}$. 
For $f\in L_{2,\pi}$ and $t\in(0,\Vert \rho \Vert_\infty)$ define
\begin{equation}  \label{eq: MK_on_level}
H_t f(x) = \int_{K(t)} f(y)\, H_t(x,\dint y), \qquad x\in K.
\end{equation}
Note that, if $x\not \in K(t)$ we have $H_t f(x) = 0$ by the convention on $H_t$, 
see \eqref{al: path_case}.
The Markov operator $H\colon L_{2,\pi}\to L_{2,\pi}$ is defined by
\[
Hf(x)=\frac{1}{\rho(x)}\int_0^{\rho(x)} H_t f(x)\, \dint t,
\]
and similarly $U\colon L_{2,\pi}\to L_{2,\pi}$ by
\[
U f(x)=\frac{1}{\rho(x)}\int_0^{\rho(x)} U_t(f)\, \dint t,
\]
where $U_t(f)=\int_{K(t)} f(x)\,U_t(\dint x)$ is a special case of \eqref{eq: MK_on_level}.
Further, for $t \in(0,\Vert \rho \Vert_\infty)$ let $L_{2,t}=L_2(K(t),U_t)$ be the space of
functions $f\colon K(t) \to \R$ with
$\norm{f}_{2,t}^2 := \langle f,f \rangle_{t} <\infty$, where
\[
\langle f,g \rangle_{t} := \int_{K(t)} f(x)\, g(x)\, U_t(\dint x)
\]
denotes the corresponding inner-product of $f,g\in L_{2,t}$.
Then, $H_t\colon L_{2,t} \to L_{2,t}$ can also be considered as Markov operator.
Define the functional
\[
S(f) = \int_K f(x)\, \pi(\dint x),\quad f\in L_{2,\pi},
\]
as operator $S\colon L_{2,\pi} \to L_{2,\pi}$ 
which maps functions to constant functions, given by their mean value.
We say $f\in L_{2,\pi}^0$ iff $f\in L_{2,\pi}$ and $S(f)=0$. 
Now the absolute spectral gap 
of a Markov kernel or Markov operator $P \colon L_{2,\pi} \to L_{2,\pi}$
is given by
\[
\gap(P) = 1 - \norm{P-S}_{L_{2,\pi} \to L_{2,\pi}} = 1- \norm{P}_{L^0_{2,\pi}\to L^0_{2,\pi}} .
\]
For details of the last equality we refer to \cite[Lemma~3.16]{Ru12}. Moreover, for the equivalence of $\gap(P)>0$ and (almost sure) geometric ergodicity we refer to \cite[Proposition~1.2]{kontoyiannis2009geometric}.
For any $t>0$ the norm $\norm{f}_{2,t}$ can also be 
considered for $f\colon K \to \mathbb{R}$. With 
this in mind we have the following relation between $\norm{f}_{2,\pi}$ and $\norm{f}_{2,t}$.
\begin{lemma} 
	For any $f\colon K\to \mathbb{R}$,
	with the notation from above, we obtain
	\begin{equation}
	\label{eq: sol_op}
	S(f) =\int_0^{\Vert \rho \Vert_\infty} U_t(f)\, \,\ell(t)\,\dint t.
	\end{equation}
	In particular, 
	\begin{align}
	\label{eq: norm}
	\norm{f}_{2,\pi}^2 & = \int_0^{\Vert \rho \Vert_\infty} \norm{f}_{2,t}^2 \,\ell(t)\,\dint t.
	\end{align}
\end{lemma}
\begin{proof}
	The assertion of \eqref{eq: norm} is a special case of \eqref{eq: sol_op}, since $S(\abs{f}^2)=\norm{f}_{2,\pi}^2$.
	By $\int_K \rho(x)\,\dint x = \int_0^{\Vert \rho \Vert_\infty} \vol_d(K(s))\, \dint s$, see in the proof of Lemma~\ref{lem: crit_rev}, 
	one obtains 
	\begin{align*}
	S(f)& 
	= \frac{\int_K f(x)\,\rho(x)\,\dint x}{\int_0^{{\Vert \rho \Vert_\infty}} \vol_d(K(s))\, \dint s}
	= \int_K \int_0^{\rho(x)} f(x) \frac{\dint t\; \dint x}{\int_0^{\Vert \rho \Vert_\infty} \vol_d(K(s))\, \dint s}
	\\
	&= \int_0^{{\Vert \rho \Vert_\infty}} \int_{K(t)} f(x)  \, \frac{\dint x}{\vol_d(K(t))}\,\ell(t) \dint t 
	= \int_0^{{\Vert \rho \Vert_\infty}} U_t(f)\,\ell(t)\, \dint t,
	\end{align*}
	which proves \eqref{eq: sol_op}.
\end{proof}

\subsection{A useful representation} \label{subsec:useful}
As in \cite[Section~3.3]{RuUl13} we derive a suitable representation 
of $H$ and $U$.
We define a $d+1$-dimensional auxiliary state space.
Let 
\[
K_\rho=\set{  (x,t)\in\R^{d+1}\mid x\in K,\, t\in(0,\rho(x))  }
\]
and let $\mu$ be the uniform distribution on $(K_\rho,\Borel(K_\rho))$, i.e.
\[
\mu(\dint(x,t)) = \frac{\dint t\, \dint x}{\vol_{d+1}(K_\rho)}.
\]
Note that $\vol_{d+1}(K_\rho) = \int_K \rho(x)\, \dint x$.
By 	$L_{2,\mu}=L_2(K_\rho,\mu)$ we denote the space of functions $f\colon K_\rho \to \R$ that satisfy $\norm{f}_{2,\mu}^2 := \langle f,f \rangle_\mu <\infty$, where 
\[
\langle f,g \rangle_\mu := \int_{K_\rho} f(x,s)\, g(x,s)\, \mu(\dint(x,s))
\]
denotes the corresponding inner-product for $f,g\in L_{2,\mu}$.
Here, similar to \eqref{eq: norm}, we have
\[
\norm{f}_{2,\mu}^2= \int_0^{\Vert \rho \Vert_\infty} \norm{f(\cdot,s)}_{2,s}^2 \ell(s)\dint s.
\]
Let $T\colon L_{2,\mu} \to L_{2,\pi}$ and $T^*\colon L_{2,\pi} \to L_{2,\mu}$ be given by
\[
Tf(x)=\frac{1}{\rho(x)} \int_0^{\rho(x)} f(x,s)\,\dint s,\quad\mbox{and}\quad
T^*f(x,s)=f(x).
\] 
Then $T^*$ is the adjoint operator of $T$,
i.e. 
for all $f\in L_{2,\pi}$ and $g\in L_{2,\mu}$ we have
\[
\scalar{f}{Tg}_{\pi} = \scalar{T^*f}{g}_\mu.
\] 
Then,
for $f\in L_{2,\mu}$ define 
\[
\widetilde{H}f(x,s)=\int_{K(s)} f(y,s)\,H_s(x,\dint y).
\]
By the stationarity of $U_s$ according to $H_s$
it is easily seen that
\begin{align*}
\norm{\widetilde{H}f}_{2,\mu}^2 
& = \int_K \int_0^{\rho(x)} \abs{\widetilde H f(x,s)}^2 \frac{\dint s\,\dint x}{\int_K \rho(y)\dint y}
 = \int_0^{\Vert \rho \Vert_\infty} \int_{K(s)} \abs{\widetilde H f(x,s)}^2 U_s(\dint x)\, \ell(s) \dint s \\
& \leq \int_0^{\Vert \rho \Vert_\infty} \int_{K(s)} \int_{K(s)} \abs{f(y,s)}^2 H_s(x,\dint y)\, U_s(\dint x)\, \ell(s) \dint s\\
& = \int_0^{\Vert \rho \Vert_\infty} \int_{K(s)} \abs{f(x,s)}^2 \, U_s(\dint x)\, \ell(s) \dint s
= \norm{f}_{2,\mu}^2.
\end{align*}
Further, define
\[
\widetilde{U}f(x,s)=\int_{K(s)} f(y,s)\,U_s(\dint y).
\]
Then,
$\widetilde{H} \colon L_{2,\mu}\to L_{2,\mu}$, $\widetilde{U}\colon L_{2,\mu} \to L_{2,\mu}$
as well as 
\[
\norm{\widetilde{H}}_{L_{2,\mu}\to L_{2,\mu}} = 1,
\quad \quad  
\norm{\widetilde{U}}_{L_{2,\mu}\to L_{2,\mu}}=1.          
\]
By the construction we have the following.
\begin{lemma} 
	\label{lem: representation}
	Let $H$, $U$, $T$, $T^*$, $\widetilde{H}$ and $\widetilde{U}$ as above. Then
	\[
	H=T \widetilde{H} T^* \quad \mbox{and} \quad U= T\widetilde{U}T^*.
	\]
\end{lemma}
Here $TT^*\colon L_{2,\pi} \to L_{2,\pi}$ satisfies $TT^*f(x)=f(x)$, i.e. $TT^*$ is the identity operator, and
$T^*T\colon L_{2,\mu} \to L_{2,\mu}$ satisfies
\[
T^*Tf(x,s) = Tf(x),
\]
i.e. it returns the average of the function $f(x,\cdot)$ over the second variable.

\section{On the spectral gap of hybrid slice samplers}
\label{sec: spectral_gap}
We start with a relation between the convergence on the slices and the convergence of $T\widetilde{H}^k T^*$ to $T\widetilde{U}T^*$ for 
increasing $k$.
\begin{lemma}  \label{lem: spec_gap_est}
	Let $k\in\N$. Then
	\begin{align*}
	&  \norm{T(\widetilde{H}^k-\widetilde{U})T^*}_{L_{2,\pi} \to L_{2,\pi}}
	\leq
	\sup_{x\in K} \left(
	\int_0^{\rho(x)}\norm{H^k_t-U_t}_{L_{2,t}\to L_{2,t}}^2 \,\frac{\dint t}{\rho(x)}\right)^{1/2}.
	\end{align*}
\end{lemma}
\begin{proof}
	First, note that $\norm{f}_{2,\pi} < \infty$ implies
	$\norm{f}_{2,t}<\infty$ 
	for  $\ell$-a.e. $t$.
	For any $k\in \N$ and $f\in L_{2,\pi}$ we have 
	\[
	(\widetilde{H}^k T^* f)(x,t) = (H^k_t f)(x)\quad  
	\mbox{and} \quad (\widetilde{U}T^*f)(x,t)=U_t(f),
	\]
	such that
	\begin{align*}
	T(\widetilde{H}^k-\widetilde{U})T^*f(x) & = 
	\int_0^{\Vert \rho \Vert_\infty} (H^k_t-U_t)f(x)\, \frac{\mathbf{1}_{K(t)}(x)}{\rho(x)} \,\dint t.
	\end{align*}
	It follows that
	\begin{align*}
	& \norm{T(\widetilde{H}^k-\widetilde{U})T^*f}_{2,\pi}^2  
	= \int_K \abs{\int_0^{\Vert \rho \Vert_\infty} (H^k_t-U_t)f(x)\, \frac{\mathbf{1}_{K(t)}(x)}{\rho(x)} \,\dint t }^2 
	\,\pi(\dint x) \\
	&  \leq  \int_K  
	\int_0^{\Vert \rho \Vert_\infty} \abs{(H^k_t-U_t)f(x)}^2\, \frac{\mathbf{1}_{K(t)}(x)}{\rho(x)} \,\dint t \,
	\frac{\rho(x)}{\int_K \rho(y) \,\dint y}\,\dint x \\
	& = \int_0^{\Vert \rho \Vert_\infty} \int_{K(t)} \abs{(H^k_t-U_t)f(x)}^2\,  \frac{\dint x}{\vol_d({K(t)})} \,
	\frac{\vol_d(K(t)) }{\int_0^{\Vert \rho \Vert_\infty} \vol_d({K(s)}) \,\dint s}\,\dint t \\
	& = \int_0^{\Vert \rho \Vert_\infty} \norm{(H^k_t-U_t)f}_{2,t}^2 \,\ell(t)\,\dint t \\
	& \leq \int_0^{\Vert \rho \Vert_\infty} \norm{H^k_t-U_t}_{L_{2,t}\to L_{2,t}}^2 \norm{f}_{2,t}^2 \,\ell(t)\,\dint t \\ 
	& = \int_0^{\Vert \rho \Vert_\infty} \int_{K(t)} \norm{H^k_t-U_t}_{L_{2,t}\to L_{2,t}}^2 \abs{f(x)}^2 \frac{\dint x}{\vol_d({K(t)})} \,
	\,      \frac{\vol_d(K(t)) }{\int_0^{\Vert \rho \Vert_\infty} \vol_d({K(s)}) \,\dint s}\,\dint t \\
	& = \int_K \int_0^{\rho(x)} \norm{H^k_t-U_t}_{L_{2,t}\to L_{2,t}}^2 \frac{\dint t}{\rho(x)} \,
	\abs{f(x)}^2\,      \frac{\rho(x) }{\int_K \rho(y) \,\dint y}\dint x \\
	& \leq  \norm{f}_{2,\pi}^2 \;	
	\sup_{x\in K} \int_0^{\rho(x)} \norm{H^k_t-U_t}_{L_{2,t}\to L_{2,t}}^2 \frac{\dint t}{\rho(x)}.
	\end{align*}
\end{proof}

\begin{remark}
	If there exists a number $\beta \in[0,1]$ such that $\norm{H_t-U_t}_{L_{2,t}\to L_{2,t}}\leq \beta$ 
	for any $t\in(0,\Vert \rho \Vert_\infty)$, then one obtains (as a consequence from the former lemma) that
	\begin{align*}
	&  \norm{T\widetilde{H}^kT^*-S}_{L_{2,\pi} \to L_{2,\pi}} 
	\leq \norm{T\widetilde{U}T^*-S}_{L_{2,\pi} \to L_{2,\pi}} + \beta^k.
	\end{align*}
	Here we employed the triangle inequality and the fact that 
	$\Vert H^k_t-U_t \Vert_{L_{2,t}\to L_{2,t}} \leq \Vert H_t-U_t \Vert_{L_{2,t}\to L_{2,t}}^k \leq \beta^k$, see for example \cite[Lemma~3.16]{Ru12}.
\end{remark}

Now a corollary follows which provides a lower bound for $\gap(T\widetilde{H}^k T^*)$.

\begin{corollary}  \label{thm: upp_est_op_norm}
	Let us assume that $\gap(U)>0$, i.e. $\norm{U-S}_{L_{2,\pi}\to L_{2,\pi}}<1$, and let us denote
	\[
	\beta_k = \sup_{x\in K} \left(
	\int_0^{\rho(x)}\norm{H^k_t-U_t}_{L_{2,t}\to L_{2,t}}^2 \,\frac{\dint t}{\rho(x)}\right)^{1/2}.
	\]
	Then 
	\begin{equation}
	\label{eq: upp_est_op_norm}
	\gap(T\widetilde{H}^kT^*) \geq \gap(U) - \beta_k.
	\end{equation}
\end{corollary}
\begin{proof}
	It is enough to prove
	\[
	\norm{T\widetilde{H}^kT^*-S}_{L_{2,\pi}\to L_{2,\pi}} \leq \norm{U-S}_{L_{2,\pi}\to L_{2,\pi}}+\beta_k.
	\]
	By $\widetilde{H}^k = \widetilde{U}+\widetilde{H}^k - \widetilde{U}$
	and Lemma~\ref{lem: spec_gap_est} we have
	\begin{align*}
	\norm{T\widetilde{H}^kT^*-S}_{L_{2,\pi}\to L_{2,\pi}}
	& = \norm{T\widetilde{U}T^*-S+T(\widetilde{H}^k-\widetilde{U})T^*}_{L_{2,\pi}\to L_{2,\pi}}\\
	& \leq \norm{U-S}_{L_{2,\pi} \to L_{2,\pi}} + \beta_k. 
	\end{align*}
\end{proof}
\begin{remark}
	If one can sample w.r.t. $U_t$ for every $t\geq 0$, then $H_t=U_t$ 
	and in the estimate of Corollary~\ref{thm: upp_est_op_norm} we obtain $\beta_k=0$ and equality in \eqref{eq: upp_est_op_norm}. 
\end{remark}
Now let us state the main theorem.  
\begin{theorem} \label{thm: low_upp_spec}
	Let us assume that for almost all  $t$ (w.r.t. $\ell$) 
	$H_t$ is positive semi-definite on $L_{2,t}$ and let 
	\[
	\beta_k = \sup_{x\in K} \left(
	\int_0^{\rho(x)}\norm{H^k_t-U_t}_{L_{2,t}\to L_{2,t}}^2 \,\frac{\dint t}{\rho(x)}\right)^{1/2}.
	\]
	Further assume that $\lim_{k\to \infty }\beta_k = 0$.
	Then
	\begin{equation}\label{eqn:main_thm}
	\frac{\gap(U)-\beta_k}{k} \leq \gap(H) \leq \gap(U), \quad k\in\N . 
	\end{equation}  
\end{theorem}

Several conclusions can be drawn from the theorem: First, under the
assumption that $\lim_{k\to \infty} \beta_k = 0,$ the LHS of 
\eqref{eqn:main_thm} implies that in the setting of the theorem, whenever the simple slice sampler
has a spectral gap, so does the hybrid version. See Section
\ref{sec: example} for examples. Second, it
also provides a quantitative bound on $\gap(H)$ given appropriate
estimates on $\gap(U)$ and $\beta_k.$ Third, the RHS of \eqref{eqn:main_thm} verifies the intuitive
result that the simple slice sampler is better than the hybrid one
(in terms of the spectral gap). 

To prove the theorem we need some further results. 

\begin{lemma} \label{lem: self_pos}
	\begin{enumerate}	
		\item \label{it: self_adj_H} 
		For any $t\in(0,\Vert \rho \Vert_{\infty})$ assume that $H_t$ is reversible with respect to $U_t$.
		Then $\widetilde{H}$ is self-adjoint on $L_{2,\mu}$.
		\item \label{it: pos_H}     Assume that for almost all $t$ (w.r.t. $\ell$) $H_t$ is positive 
		semi-definite on $L_{2,t}$, 
		i.e. for all $f\in L_{2,t}$ holds $\scalar{H_t f}{f}_t \geq 0$.
		Then $\widetilde{H}$ is positive semi-definite on $L_{2,\mu}$. 
	\end{enumerate}
\end{lemma}
\begin{proof}
	Note that $\norm{f}_{2,\mu}<\infty$ 
	implies $\norm{f(\cdot,t)}_{2,t}<\infty$ for almost all $t$ (w.r.t. $\ell$).\\[0.5ex]
	\textbf{To \ref{it: self_adj_H}.:}
	Let $f,g\in L_{2,\mu}$ then we have to show that 
	\[
	\scalar{\widetilde{H}f}{g}_{\mu} = \scalar{f}{\widetilde{H}g}_\mu.
	\]
	Note that for $f,g\in L_{2,\mu}$ we have for almost all $t$, 
	by the reversibility of $H_t$, that
	\[
	\scalar{H_t f(\cdot,t)}{g(\cdot,t)}_t = \scalar{ f(\cdot,t)}{H_tg(\cdot,t)}_t.
	\]
	By
	\begin{align*}
	\scalar{\widetilde{H}f}{g}_{\mu}
	& = \int_{K_\rho} \widetilde{H}f(x,t) g(x,t) \,\mu(\dint(x,t))  \\
	& = \int_K \int_0^{\rho(x)} \int_{K(t)} f(y,t)\,H_t(x,\dint y) g(x,t)\, \frac{\dint t\, \dint x}{\vol_{d+1}(K_\rho)} \\
	& = \int_0^{\Vert \rho \Vert_\infty} \int_{K(t)} \int_{K(t)} f(y,t)\,H_t(x,\dint y) g(x,t)\, U_t(\dint y)\, \ell(t)\, \dint t \\
	& = \int_0^{\Vert \rho \Vert_\infty}   \scalar{H_t f(\cdot,t)}{g(\cdot,t)}_t \, \ell(t)\, \dint t
	\end{align*}
	the assertion of \ref{it: self_adj_H}. is proven. \\[0.5ex]	
	\textbf{To \ref{it: pos_H}.:}
	We have to prove for all $f\in L_{2,\mu}$ that
	\[
	\scalar{\widetilde{H}f}{f}_{\mu} = \int_{K_\rho} \widetilde{H}f(x,t) f(x,t) \, \mu(\dint(x,t))
	\geq 0.
	\]
	Note that for $f\in L_{2,\mu}$ we have for almost all $t$ that 
	\[
	\scalar{H_t f(\cdot,t)}{f(\cdot,t)}_t\geq 0.
	\]
	By the same computation as in \ref{it: self_adj_H}.
	we obtain that the positive semi-definite\-ness of $H_t$ carries over to $\widetilde{H}$.
\end{proof}

The statement and proofs of the following lemmas follow closely the line of arguments of
\cite{Ul12,Ul14} and essentially use \cite[Lemma~\ref{lem: tech_lemma} and Lemma~\ref{lem: tech_lem_2}]{Ul14}.
We provide alternative proofs of the aforementioned employed lemmas in Appendix~\ref{app: tech_lemmas}.

\begin{lemma}  \label{lem: mon}
	Let $\widetilde{H}$ be positive semi-definite on $L_{2,\mu}$. Then
	\begin{equation} \label{eq: mon}
	\norm{T \widetilde{H}^{k+1} T^*	-S}_{L_{2,\pi}\to L_{2,\pi}} \leq  
	\norm{T \widetilde{H}^k T^*-S}_{L_{2,\pi}\to L_{2,\pi}},\quad k\in\N.
	\end{equation}
	Further, if  
	\[
	\beta_k = \sup_{x\in K} \left(
	\int_0^{\rho(x)}\norm{H^k_t-U_t}_{L_{2,t}\to L_{2,t}}^2 \,\frac{\dint t}{\rho(x)}\right)^{1/2}
	\]
	and $\lim_{k\to \infty }\beta_k = 0$.
	Then  \[
	\norm{U-S}_{L_{2,\pi}\to L_{2,\pi}} \leq  \norm{H-S}_{L_{2,\pi}\to L_{2,\pi}}.
	\]
\end{lemma}
\begin{proof}
	Let  $S_1 \colon L_{2,\mu} \to L_{2,\pi}$ 
	and the adjoint $S_1^* \colon L_{2,\pi} \to L_{2,\mu}$ be given by
	\[
	S_1(f) = \int_{K_\rho} f(x,s)\, \mu(\dint(x,s))
	\quad \mbox{and} \quad
	S_1^*(g)= \int_K g(x) \, \pi(\dint x).
	\]
	Thus, $\scalar{S_1 f}{g}_\pi = \scalar{f}{S_1^* g}_\mu$. 
	Furthermore observe that $S_1S_1^*=S$.
	Let $R=T-S_1$ and note that $RR^*=I-S$, with identity $I$, and $RR^*=(RR^*)^2$. 
	Since $RR^*\not = 0$ and the projection property $RR^*=(RR^*)^2$ 
	one gets $\norm{RR^*}_{L_{2,\pi}\to L_{2,\pi}}=1$.
	We have
	\begin{align*}
	R \widetilde{H}^k R^* & = (T-S_1) \widetilde{H}^k (T^*-S_1^*) \\
	& = T\widetilde{H}^kT^*-T\widetilde{H}^k S_1^*- S_1 \widetilde{H}^k T^*+ S_1 \widetilde{H}^k S_1^* 
	= T\widetilde{H}^kT^*-S.
	\end{align*}
	Further $\norm{S_1 \widetilde{H} S_1^*}_{L_{2,\mu}\to L_{2,\mu}}\leq 1$. 
	Then, by Lemma~\ref{lem: tech_lemma} it follows that
	\[
	\norm{R \widetilde{H}^{k+1} R^*}_{L_{2,\pi}\to L_{2,\pi}} 
	\leq \norm{R \widetilde{H}^{k} R^*}_{L_{2,\pi}\to L_{2,\pi}}
	\]
	and the proof of \eqref{eq: mon} is completed.\\ 
	By Lemma~\ref{lem: spec_gap_est} we obtain
	$\norm{T(\widetilde{H}^k-\widetilde{U})T^*}_{L_{2,\pi}\to L_{2,\pi}} \leq \beta_k$ 
	and by \eqref{eq: mon} as well as Lemma~\ref{lem: representation} we obtain 
	\[
	\norm{T \widetilde{H}^{k} T^*-S}_{L_{2,\pi}\to L_{2,\pi}} \leq  \norm{H-S}_{L_{2,\pi}\to L_{2,\pi}},\quad k\in\N.
	\]  
	This implies by triangle inequality that
	\[
	\lim_{k\to\infty} \norm{T \widetilde{H}^{k} T^*-S}_{L_{2,\pi}\to L_{2,\pi}} 
	= \norm{U-S}_{L_{2,\pi}\to L_{2,\pi}}
	\]
	and the assertion is proven.
\end{proof}
\begin{lemma}
	Let $\widetilde{H}$ be positive semi-definite on $L_{2,\mu}$. Then 
	\begin{equation} \label{eq: k_indep}
	\norm{H-S}_{L_{2,\pi}\to L_{2,\pi}}^{k} 
	\leq  \norm{T \widetilde{H}^k T^*-S}_{L_{2,\pi}\to L_{2,\pi}}, 
	\end{equation}
	for any $k\in\N$.
\end{lemma}
\begin{proof}
	As in the proof of Lemma~\ref{lem: mon} we use $R \widetilde{H}^k R^* = T\widetilde{H}^kT^*-S$ 
	to reformulate the assertion. 
	It remains to prove that
	\[
	\norm{R\widetilde{H}R^*}_{L_{2,\pi}\to L_{2,\pi}}^{k} 
	\leq \norm{R \widetilde{H}^k R^*}_{L_{2,\pi}\to L_{2,\pi}}.
	\]
	Recall that $RR^*$ is a projection and satisfies $\norm{RR^*}_{L_{2,\pi}\to L_{2,\pi}} = 1$.
	By Lemma~\ref{lem: tech_lem_2} the assertion is proven.
\end{proof}

Now we turn to the proof of Theorem~\ref{thm: low_upp_spec}.
\begin{proof}[Proof of Theorem~\ref{thm: low_upp_spec}]
	By Lemma~\ref{lem: self_pos} we know that $\widetilde{H}\colon L_{2,\mu} \to L_{2,\mu}$ 
	is self-adjoint and positive semi-definite.
	By Lemma~\ref{lem: mon} we have
	\[
	\norm{U-S}_{L_{2,\pi}\to L_{2,\pi}} \leq  \norm{H-S}_{L_{2,\pi}\to L_{2,\pi}}.
	\]
	By Theorem~\ref{thm: upp_est_op_norm} we have for any $k\in\N$ that
	\begin{equation}  \label{eq: L_2_k_est}
	\norm{T\widetilde{H}^kT^*-S}_{L_{2,\pi}\to L_{2,\pi}} \leq \norm{U-S}_{L_{2,\pi}\to L_{2,\pi}}+\beta_k.
	\end{equation}
	Then
	\begin{align*}
	\norm{U-S}_{L_{2,\pi}\to L_{2,\pi}}
	& \underset{\eqref{eq: L_2_k_est}} {\geq} 
	\norm{T \widetilde{H}^k T^*-S}_{L_{2,\pi}\to L_{2,\pi}} - \beta_k \\
	& \underset{\eqref{eq: k_indep}} {\geq} \norm{H-S}_{L_{2,\pi}\to L_{2,\pi}}^{k} - \beta_k\\
	& \geq 1- k\,(1-\norm{H-S}_{L_{2,\pi}\to L_{2,\pi}}) - \beta_k\\
	& = 1-k\,\gap(H) - \beta_k,
	\end{align*}
	where we applied a version of Bernoulli's inequality, 
	i.e. $1-x^n\leq n (1-x)$ for $x\geq 0$ and $n\in\N$.
	Thus,
	\[
	\frac{\gap(U)-\beta_k}{k} \leq \gap(H)
	\]
	and the proof is completed.
\end{proof}

\section{Applications} \label{sec: example}
In this section we apply Theorem~\ref{thm: low_upp_spec} under different assumptions
with different Markov chains on the slices. We provide a criterion of geometric ergodicity of these hybrid
slice samplers by showing that there is a spectral gap whenever the simple slice sampler 
has a spectral gap.

First we consider a class of bimodal densities in a $1$-dimensional setting.
We study a stepping-out shrinkage slice sampler, suggested in \cite{Ne03}, 
which is explained in Algorithm~\ref{fig: step_shrink}.

Then we consider a hybrid slice sampler which performs a hit-and-run 
step on the slices in a $d$-dimensional
setting. Here we impose very weak assumptions on the unnormalised densities.
The drawback is that an implementation of this algorithm might be difficult.

Motivated by this difficulty we study a combination
of the previous sampling procedures on the slices. 
The resulting hit-and-run stepping-out shrinkage slice sampler is presented in Algorithm~\ref{fig: har_st_sh}.
Here we consider a class of bimodal densities in a $d$-dimensional setting.

\subsection{Stepping-out and shrinkage procedure}
Let
$w>0$ be a parameter and $\rho \colon \R \to (0,\infty)$ be
an unnormalised density.
We say $\rho\in \mathcal{R}_{w}$ if
the following conditions are satisfied:
There exist $t_1,t_2\in(0,\norm{\rho}_{\infty})$
with $t_1\leq t_2$ such that 
\begin{enumerate}
	\item\label{en: conv_t_1}  for all 
	$t\in(0,t_1)\cup [t_2,\norm{\rho}_{\infty})$
	the level set $K(t)$ is an interval;
	\item\label{en: cup_t_1_2} for all $t\in[t_1,t_2)$ 
	there are disjoint intervals $K_1(t),K_2(t)$
	with strictly positive Lebesgue measure,
	such that 
	\[
	K(t)=K_1(t)\cup K_2(t)
	\]
	and for all $\eps>0$ holds $K_i(t+\eps)\subseteq K_i(t)$ for $i=1,2$. For convenience we set $K_i(t) =\emptyset$ for $t\not \in [t_1,t_2)$.
	\item\label{en: delta_w} for all $t\in(0,\norm{\rho}_{\infty})$ we assume $\delta_t < w$ where
	\[
	\delta_t := \begin{cases}
	\inf_{r\in K_1(t),\; s\in K_2(t)} \abs{r-s} & t\in [t_1,t_2) \\
	0 & \mbox{otherwise}.
	\end{cases} 
	\]
\end{enumerate} 
The next result shows that certain bimodal densities belong to $\mathcal{R}_w$.
\begin{lemma}  \label{lem: max}
	Let $\rho_1 \colon \R \to (0,\infty)$ 
	and $\rho_2 \colon \R \to (0,\infty)$ be unnormalised
	density functions. 
	Let us assume that $\rho_1$, $\rho_2$ are lower semi-continuous and quasi-concave, i.e.
	the level sets are open intervals,
	and 
	\[
	\inf_{r \in \arg \max \rho_1,\, s\in \arg\max \rho_2} \abs{r-s}<w.
	\]
	Then $\rho_{\max} :=\max\{\rho_1,\rho_2\} \in \mathcal{R}_w$.
\end{lemma}
\begin{proof}
	For $t\in (0,\norm{\rho_{\max}}_{\infty})$ let
	$K_{\rho_{\max}}(t)$, $K_{\rho_1}(t)$ and $K_{\rho_2}(t)$ be the level sets 
	of $\rho_{\max}$, $\rho_1$ and $\rho_2$ of level $t$. Note that 
	\[
	K_{\rho_{\max}}(t) = K_{\rho_1}(t) \cup K_{\rho_2}(t).
	\]
	With the choice
	\begin{align*}
	t_1 & = \inf \{ t\in (0,\norm{\rho_{\max}}_{\infty}) : K_{\rho_1}(t) \cap K_{\rho_2}(t) = \emptyset \},\\ 
	t_2 & = \min\{ \norm{\rho_1}_{\infty}, \norm{\rho_2}_{\infty} \}
	\end{align*}
	we have \ref{en: conv_t_1}. and \ref{en: cup_t_1_2}.
	Observe that $\arg \max \rho_i \subseteq K_{\rho_i}(t)$ for $i=1,2$, which yields
	\begin{align*}
	\inf_{r\in K_{\rho_1}(t),\, s\in K_{\rho_2}(t) } \abs{r-s} 
	& \leq \inf_{r\in \arg \max \rho_1,\,s\in \arg \max \rho_2 } \abs{r-s}
	< w.
	\end{align*}
\end{proof}

In \cite{Ne03} a stepping-out and shrinkage procedure is suggested
for the transitions on the level sets. The procedures are explained 
in Algorithm~\ref{fig: step_shrink}
where a single transition from the resulting 
hybrid slice sampler 
from $x$ to $y$ is presented.
\begin{algorithm}
	\label{fig: step_shrink}
	A hybrid slice sampling transition of the stepping-out and shrinkage procedure from $x$ to $y$, i.e. input $x$ and output $y$.
	The stepping-out procedure has input $x$ (current state), $t$ (chosen level),
	$w>0$ (step size parameter from $\mathcal{R}_w$) and outputs an interval $[L,R]$.
	The shrinkage procedure has input $[L,R]$ and output $y$:
	\begin{enumerate}
		\item Choose a level $t \sim \mathcal{U}(0,\rho(x))$;
		\item Stepping-out with input $x,t,w$ outputs an interval $[L,R]$:
		\begin{enumerate}
			\item Choose $u \sim \mathcal{U}[0,1]$. Set $L=x-u w$ and $R=L+w$;
			\item Repeat until $t \geq \rho(L)$, i.e. $L \not \in K(t)$: Set $L=L-w$;
			\item Repeat until $t \geq \rho(R)$, i.e. $R \not \in K(t)$: Set $R=R+w$;
		\end{enumerate}
		\item Shrinkage procedure with input $[L,R]$ outputs $y$:
		\begin{enumerate}
			\item Set $\bar{L}=L$ and $\bar{R}=R$;
			\item Repeat:
			\begin{enumerate}
				\item   Choose $v\sim \mathcal{U}[0,1]$ and set $y=\bar{L}+ v (\bar{R}-\bar{L})$;
				\item   If $y  \in K(t)$ then return $y$ and exit the loop;
				\item   If $y<x$ then set $\bar{L}=y$, else $\bar{R}=y$.
			\end{enumerate}
		\end{enumerate}
	\end{enumerate}
\end{algorithm}

For short we write $\abs{K(t)}=\vol_1(K(t))$  and for $t\in(0,\Vert\rho\Vert_\infty)$ we set
\[
\gamma_t := \frac{(w-\delta_t)}{w} \frac{\abs{K(t)}}{(\abs{K(t)}+\delta_t)}.
\]
Now we provide useful results to apply Theorem~\ref{thm: low_upp_spec}.
\begin{lemma} \label{lema: 1_dim}
	Let $\rho\in \mathcal{R}_w$ with $t_2> 0$ satisfying 
	\ref{en: conv_t_1}. and \ref{en: cup_t_1_2}. of the definition of $\mathcal{R}_w$. Moreover, let $t\in(0,\Vert \rho \Vert_\infty)$.
	\begin{enumerate}
		\item \label{it: repres_H_t}
		Then, the transition kernel $H_t$ of the stepping-out and shrinkage slice sampler from Algorithm~\ref{fig: step_shrink} takes the form
		\begin{align*}
		& H_t(x,A) = 
		\gamma_t\,	U_t(A) + (1 - \gamma_t) 
		\left[\mathbf{1}_{K_1(t)}(x) U_{t,1}(A) + \mathbf{1}_{K_2(t)}(x) U_{t,2}(A) \right],
		\end{align*}
		with $x\in \mathbb{R}$, $A\in \mathcal{B}(\mathbb{R})$ and  
		\[
		U_{t,i} (A) = \begin{cases} \frac{\abs{K_i(t)\cap A}}{\abs{K_i(t)}}, & t\in[t_1,t_2), \\ 
		0, & t\in (0,t_1) \cup [t_2,\Vert \rho \Vert_\infty),
		\end{cases}
		\]
		for $i=1,2$, i.e. in the case $t\in[t_1,t_2) $ we have that 
		$U_{t,i}$ denotes the uniform distribution in $K_{i}(t)$.
		(For $t\in(0,t_1)\cup [t_2,\norm{\rho}_{\infty})$ we have $H_t = U_t$ since $\delta_t=0$ yields $\gamma_t=1$.) 
		\item \label{it: ex1_pos}
		The transition kernel $H_t$ is reversible and
		induces a positive semi-definite operator, i.e. for any $f\in L_{2,t}$ holds
		$\scalar{H_t f}{f}_{t}\geq 0$.
		\item \label{it: ex1_beta}
		Then
		$
		\norm{H_t-U_t}_{L_{t,2}\to L_{t,2}} = 1 - 
		\gamma_t
		$
		and
		\begin{equation}
		\label{eq: beta_express_1st_ex}
		\beta_k \leq \left(\frac{1}{t_2} \int_{0}^{t_2}
		\left(1 - 
		\gamma_t
		\right)^{2k}
		\dint t\right)^{1/2}, \quad k\in\N.
		\end{equation}
	\end{enumerate}
\end{lemma}
\begin{proof}
	\textbf{To \ref{it: repres_H_t}.:} 
	For $t\in (0,t_1)\cup [t_2,\Vert \rho \Vert_{\infty})$ the stepping-out procedure returns an interval that contains $K(t)$ entirely. Then, since $K(t)$ is also an interval, the shrinkage scheme returns a sample w.r.t. $U_t$ in $K(t)$.
	
	For $t\in [t_1,t_2)$, $i\in \{1,2\}$ and $x\in K_i(t)$, within the stepping-out procedure with probability $(w-\delta_t)/w$ an interval that contains $K(t)=K_1(t) \cup K_2(t)$ and with probability $1-(w-\delta_t)/w$ an interval that contains $K_i(t)$ but not $K(t)\setminus K_i(t)$ is returned. We distinguish these cases:\\[0.25ex]
	\emph{Case 1: `$K(t)$ contained in the stepping-out output':}
	
	Then, the shrinkage scheme returns with probability $\vert K(t) \vert/(\vert K(t) \vert+\delta_t)$ a sample w.r.t. $U_t$ and with probability $1-\vert K(t) \vert/(\vert K(t) \vert+\delta_t)$ a sample w.r.t. $U_{t,i}$.\\[0.25ex]
	\emph{Case 2: `$K_i(t)$, but not $K(t) \setminus K_i(t)$, contained in the stepping-out output':}
	
	Then, the shrinkage scheme returns with probability $1$ a sample w.r.t. $U_{t,i}$.\\[0.25ex]
	In total,  for $x\in K_i(t)$ we obtain
	\begin{align*}
	H_t(x,A) & = \frac{(w-\delta_t)}{w} \Big[\frac{\vert K(t)\vert}{\vert K(t) \vert + \delta_t} \; U_t(A) +
	\big(1-\frac{\vert K(t)\vert}{\vert K(t) \vert + \delta_t}\big)U_{t,i}(A) \Big] \\
	& \qquad \qquad
	+\big(1-\frac{(w-\delta_t)}{w} \big) U_{t,i}(A)\\
	& = \gamma_t U_t (A) + (1-\gamma_t) U_{t,i}(A),
	\end{align*}
	where we emphasize that for $t\in (0,t_1)\cup [t_2,\Vert \rho \Vert_{\infty})$ follows $\gamma_t=1$ (since $\delta_t=0$), such that $H_t(x,A)$ coincides with $U_t(A)$.\\[0.25ex]
	\textbf{To~\ref{it: ex1_pos}.:} For $A,B\in\mathcal{B}(\mathbb{R})$ we have
	\begin{align*}
	&	\int_A H_t(x,A) \,U_t(\dint x) 
	= \gamma_t\, U_t(B) U_t(A)\\
	&\qquad \qquad + (1-\gamma_t)\int_A\big[\mathbf{1}_{K_1(t)}(x) U_{t,1}(B)+\mathbf{1}_{K_2(t)}(x) U_{t,2}(B)\big] U_t(\dint x)\\
	& = \gamma_t\, U_t(B) U_t(A) + (1-\gamma_t) 
	\Big[ \frac{\vert K_1(t) \vert}{\vert K(t )\vert} U_{t,1}(A) U_{t,1}(B) + \frac{\vert K_2(t)\vert}{\vert K(t)\vert} U_{t,2}(A) U_{t,2}(B) \Big],
	\end{align*}
	which is symmetric in $A,B$ and therefore implies the claimed reversibility w.r.t. $U_t$. Similarly, we have
	\begin{align}
	\label{al: positive_repres}
	\scalar{H_t f}{f}_t & 
	= \gamma_t\; U_t(f)^2 +
	(1 - \gamma_t ) 
	\left[ \frac{\abs{K_1(t)}}{\abs{K(t)}}\; U_{t,1}(f)^2 
	+ \frac{\abs{K_2(t)}}{\abs{K(t)}}\;U_{t,2}(f)^2 \right]\geq 0,
	\end{align}
	where $U_{t,i}(f)$ denotes the integral of $f$ w.r.t. $U_{t,i}$ 
	for $i=1,2$, which proves the positive semi-definiteness.\\[0.25ex]
	\textbf{To \ref{it: ex1_beta}.:}
	By virtue of \cite[Lemma~3.16]{Ru12}, the reversibility (equivalently self-adjointness) of $H_t$ and \cite[Theorem~V.5.7]{We05} we have
	\begin{equation}
	\label{eq: norm_sup}
	\Vert H_t - U_t \Vert_{L_{2,t}\to L_{2,t}} = \sup_{\Vert f\Vert_{2,t} \leq 1,\, U_t(f)=0} \vert \langle H_tf,f\rangle \vert
	= \sup_{\Vert f\Vert_{2,t} \leq 1,\, U_t(f)=0} \langle H_tf,f\rangle,
	\end{equation}
	where the last equality follows by the positive semi-definiteness. Observe that for any $f\in L_{2,s}$ with $s\in [t_1,t_2)$ we have by $U_{s,i}(f)^2 \leq U_{s,i}(f^2)$ for $i=1,2$ that
	\begin{align*}
	\frac{\abs{K_1(s)}}{\abs{K(s)}}\; U_{s,1}(f)^2 
	+ \frac{\abs{K_2(s)}}{\abs{K(s)}}\;U_{s,2}(f)^2 
	\leq 
	\frac{\abs{K_1(s)}}{\abs{K(s)}}\; U_{s,1}(f^2) 
	+ \frac{\abs{K_2(s)}}{\abs{K(s)}}\;U_{s,2}(f^2) = \Vert f \Vert_{2,s}^2.  
	\end{align*}
	Therefore, the equality in \eqref{al: positive_repres} yields
	\begin{equation}
	\label{eq: to_be_attained}
	\Vert H_t - U_t \Vert_{L_{2,t}\to L_{2,t}} \leq \sup_{\Vert f\Vert_{2,t} \leq 1,\, U_t(f)=0} (1-\gamma_t) \Vert f \Vert_{2,t}^2 = 1-\gamma_t.
	\end{equation}
	For $t\in (0,t_1)\cup [t_2,\Vert \rho \Vert_\infty)$ by $H_t=U_t$ and $1-\gamma_t=0$ we have an equality. For $t\in [t_1,t_2)$ with
	\[
	h(x) = \frac{\abs{K(t)}}{\abs{K_1(t)}}\; \mathbf{1}_{K_1(t)}(x)
	- \frac{\abs{K(t)}}{\abs{K_2(t)}}\; \mathbf{1}_{K_2(t)}(x)
	\]
	the upper bound of \eqref{eq: to_be_attained} is attained for $f=h/\Vert h \Vert_{2,t}$ in the supremum expression of \eqref{eq: norm_sup}.
	
	We turn to the verification of \eqref{eq: beta_express_1st_ex}: For $t\in (t_2,\norm{\rho}_{\infty}]$ we have
	$1-\gamma_t=0$ and for $t\in(0,t_2)$ the function $1-\gamma_t$ is increasing
	which also yields that $t\mapsto (1-\gamma_t)^{2k}$ is increasing on $(0,t_2)$ for any $k\in\mathbb{N}$. By \cite[Lemma~3.16]{Ru12} we obtain
	\[
	\Vert H^k_t - U_t \Vert_{L_{2,t} \to L_{2,t}} \leq (1-\gamma_t)^{k}.
	\]
	Consequently, we have
	\begin{equation}
	\label{eq: estimate}
	\beta_k \leq \sup_{r\in (0,t_2)} \frac{1}{r} \int_0^r (1-\gamma_t)^{2k} \dint t.
	\end{equation}
	Further, note that for $a\in(0,\infty)$, 
	any increasing function $g\colon (0,a) \to \R$ and $p,q\in (0,a)$
	with $p\leq q$ holds
	\begin{equation}
	\label{eq: aux_ineq}
	\frac{1}{p} \int_0^p g(t)\, \dint t \leq \frac{1}{q} \int_0^q g(t)\,\dint t.
	\end{equation}
	The former inequality can be verified by showing that the function $p\mapsto G(p)$ for $p\geq0$
	with $G(p)=\frac{1}{p} \int_0^p g(t)\, \dint t$ satisfies $G'(p)\geq 0$.
	Applying \eqref{eq: aux_ineq} with $g(t)=(1-\gamma_t)^{2k}$
	in combination with \eqref{eq: estimate} yields \eqref{eq: beta_express_1st_ex}.
\end{proof}

By Theorem~\ref{thm: low_upp_spec} and the previous lemma we have the 
following result.

\begin{corollary}
	For any $\rho \in \mathcal{R}_w$ the stepping-out and shrinkage slice sampler 
	has a spectral gap
	if and only if the simple slice sampler has a spectral gap.
\end{corollary}
\begin{remark}
	We want to discuss two extremal situations: 
	\begin{itemize}
		\item Consider densities $\rho \colon \mathbb{R} \to (0,\infty)$ that are lower semi-continuous and quasi-concave, i.e. the level sets are open intervals. Loosely formulated we assume to have uni-modal densities. Then, for any $w>0$ we have $\rho\in\mathcal{R}_w$ (just take $t_1=t_2$ arbitrarily) and $\delta_t=0$ for all $t\in(0,\Vert \rho \Vert_\infty)$. Hence, $H_t=U_t$ for all $t\in(0,\Vert \rho \Vert_\infty)$ and Algorithm~\ref{fig: step_shrink} provides an effective implementation of simple slice sampling. 
		\item Assume that $\rho\colon \mathbb{R} \to (0,\infty)$ satisfies \ref{en: conv_t_1}. and \ref{en: cup_t_1_2}. from the definition of $\mathcal{R}_w$ for some $w>0$, but for any $t\in(0,\Vert \rho \Vert_\infty)$ we have $\delta_t\geq w$. In this setting our theory is not applicable and it is clear that the corresponding Markov chain does not work well, since it is not possible to get from one part of the support to the other one. 
	\end{itemize}
\end{remark}

\subsection{Hit-and-run slice sampler}

The idea is to combine the hit-and-run algorithm with slice sampling.
We ask whether a spectral gap of simple slice sampling implies a spectral gap 
of this combination.
The hit-and-run algorithm was proposed by Smith \cite{Sm84}. It is well studied, 
see for example \cite{BeRoSm93,DiFr97,KiSmZa11,KaSm98,Lo99,LoVe06,RuUl13}, and used for numerical integration,
see \cite{Ru12,Ru13}. We define the setting and the transition kernel
of hit-and-run.

We consider a $d$-dimensional state space $K\subseteq \R^d$ 
and $\rho\colon K \to (0,\infty)$ is an unnormalised
density. 
We denote the diameter of a level set by 
\[
\diam(K(t))=\sup_{x,y\in K(t)} \abs{x-y}
\]
with the
Euclidean norm $\abs{\cdot}$.
We impose the following assumption.
{ \assumption
		\label{ass: reg_level_sets}
		The limit $\kappa:=\lim_{t\downarrow 0} \frac{\vol_d(K(t))}{\diam(K(t))^d}$ exists and there are numbers $c,\varepsilon\in (0,1]$, such that
		\begin{equation}  \label{eq: limit_cond}
		\inf_{t\in (0,\varepsilon)}  \frac{\vol_d(K(t))}{\diam(K(t))^d} =
		c>0.
		\end{equation}
}\\
	Note that under Assumption~\ref{ass: reg_level_sets} we always have $\kappa \geq c$.
	If $K$ is bounded, has positive Lebesgue measure and $\inf_{x\in K} \rho(x)>0$, then Assumption~\ref{ass: reg_level_sets} is satisfied with $\kappa=c$. 
	Moreover, for instance, the density of a standard normal distribution satisfies Assumption~\ref{ass: reg_level_sets}
	with unbounded $K$ where also $c=\kappa$. However, the following example indicates that this is not always the case.
	\begin{example} \label{ex: Ass2}
		Let $K=(0,1)^2$ and $\rho(x_1,x_2)=2-x_1-x_2$. Then, for $t\in(0,1]$ we have
		\[
		K(t) = \{(x_1,x_2)\in (0,1)^2\colon x_2\in (0,\min\{1,2-t-x_1\})\},
		\]
		such that $\vol_2(K(t))=1-t^2/2$. Moreover, the fact that $\{(\alpha,1-\alpha)\colon \alpha\in (0,1)\}\subseteq K(t)$ yields $\diam(K(t))=\sqrt{2}$, such that for $\varepsilon=1$ we have 
		$c=1/4$ 
		and $\kappa = 1/2$.
	\end{example}
In general, we
	consider Assumption~\ref{ass: reg_level_sets} as weak regularity requirement,
since there is no
condition on the level sets and also no condition on the modality.

Let $S_{d-1}$ be the Euclidean unit sphere
and $\sigma_d = \vol_{d-1}(S_{d-1})$.
A transition from $x$ to $y$ by hit-and-run on the level set $K(t)$ works as follows: 
\begin{enumerate}
	\item Choose $\theta \in S_{d-1}$ uniformly distributed;
	\item Choose $y$ according to the uniform distribution on the line
	$x+ r \theta$ intersected with $K(t)$.
\end{enumerate}
This leads to	
\[
H_t(x,A) 
= 
\int_{S_{d-1}} \int_{L_t(x,\theta)} \mathbf{1}_A(x+s\theta) \frac{\dint s}{\vol_1(L_t(x,\theta))}
\,
\frac{\dint \theta}{\sigma_d}
\]
with 
\[
L_t(x,\theta) = \{ r \in \R  \mid x+r \theta \in K(t)  \}.
\]
The hit-and-run algorithm is reversible and induces
a positive-semidefinite operator on $L_{2,t}$, see \cite{RuUl13}.
The following property is well known, see for example \cite{DiFr97}.

\begin{proposition}
	For $t\in (0,\norm{\rho}_{\infty})$, $x\in K(t)$ and $A\in \mathcal{B}(K)$ we have
	\begin{equation}  \label{eq: repr_har}
	H_t(x,A) = 
	\frac{2}{\sigma_d}
	\int_A \frac{\dint y}{\abs{x-y}^{d-1} \vol_1(L(x,\frac{x-y}{\abs{x-y}}))}
	\end{equation}
	and 
	\begin{equation}  \label{eq: gap_har}
	\norm{H_t - U_t}_{L_{2,t} \to L_{2,t}} \leq 1 - \frac{2}{\sigma_d}\, \frac{ \vol_d(K(t))}{\diam(K(t))^d}.
	\end{equation}
\end{proposition}

\begin{proof}
	The representation of $H_t$ stated in \eqref{eq: repr_har} is well known, see for example \cite{DiFr97}.
	From this we have for any $x\in K(t)$ that
	\[
	H_t(x,A) \geq \frac{2}{\sigma_d} \frac{\vol_d(K(t))}{\diam(K(t))^d} \cdot \frac{\vol_d(K(t)\cap A)}{\vol_d(K(t))}.
	\]
	which means that the whole state space $K(t)$ is a small set.	 By \cite{MeTw09}
	we have uniform ergodicity and by \cite[Proposition~3.24]{Ru12} we obtain 
	\eqref{eq: gap_har}.
\end{proof}

Further, we obtain the following helpful result.
\begin{lemma} \label{lem: har_conv_slice}
	Under Assumption~\ref{ass: reg_level_sets} we have
	with 
	\[
	\beta_k 
	= \sup_{x\in K} \left( \int_0^{\rho(x)}\norm{H^k_t-U_t}_{L_{2,t}\to L_{2,t}}^2 \,\frac{\dint t}{\rho(x)}\right)^{1/2}
	\]
	that
	$
	\lim_{k \to \infty} \beta_k
	= 0.
	$
\end{lemma}
\begin{proof}
	By \eqref{eq: gap_har} and \cite[Lemma~3.16]{Ru12}, in combination with reversibility (equivalently self-adjointness) of $H_t$, holds 
	\begin{equation}  \label{eq: unif_erg_har}
	\norm{H_t^k - U_t}_{L_{2,t}\to L_{2,t}}^2 \leq \left( 1-\frac{2}{\sigma_d} \frac{\vol_d(K(t))}{\diam(K(t))^d} \right)^{2k}.
	\end{equation}
	Let $g_k\colon [0,\Vert \rho \Vert_\infty) \to [0,1]$ be given by
	\[
	g_k(t) = 
	\begin{cases}
	\left( 1-\frac{2}{\sigma_d} \frac{\vol_d(K(t))}{\diam(K(t))^d} \right)^{2k} & t\in(0,\Vert \rho \Vert_\infty),\\
	\left( 1-\frac{2\, \kappa}{\sigma_d} \right)^{2k} & t=0,
	\end{cases}
	\]
	which is
	the continuous extension at zero of the  
	upper bound of \eqref{eq: unif_erg_har}
	with $\kappa\geq c\in (0,1]$ from Assumption~\ref{ass: reg_level_sets}. Note that $\lim_{k\to \infty} g_k(t) = 0$ for all $t\in[0,\Vert \rho \Vert_\infty)$ and
	$
	\beta_k \leq \sup_{r\in (0,\norm{\rho}_{\infty}]}\left( \frac{1}{r} \int_0^r g_k(t)\,\dint t \right)^{1/2}.
	$
	Considering the continuous function
	\[
	h_k(r) = 
	\begin{cases}
	\frac{1}{r} \int_0^r g_k(t)\dint t & r\in (0,\Vert \rho\Vert_\infty],\\
	g_k(0) & r=0,
	\end{cases}
	\]
	the supremum can be replaced by a maximum over $r\in[0,\norm{\rho}_{\infty}]$
	which is attained, say for $r^{(k)} \in [0,\norm{\rho}_{\infty}]$, i.e. $\beta_k\leq h_k(r^{(k)})^{1/2}$.
	Define 
	\begin{align*}
	(r_0^{(k)})_{k\in\mathbb{N}} & := \{ r^{(k)} \mid r^{(k)} =0, k\in\mathbb{N} \} \subseteq (r^{(k)})_{k\in\mathbb{N}},\\
	(r_1^{(k)})_{k\in\mathbb{N}} & := \{ r^{(k)} \mid r^{(k)} \in (0,\varepsilon), k\in\mathbb{N} \} \subseteq (r^{(k)})_{k\in\mathbb{N}},\\
	(r_2^{(k)})_{k\in\mathbb{N}} & := \{ r^{(k)} \mid r^{(k)} \geq \varepsilon, k\in\mathbb{N} \} \subseteq (r^{(k)})_{k\in\mathbb{N}}.
	\end{align*}
	W.l.o.g. we assume that $(r_0^{(k)})_{k\in\mathbb{N}} \not = \emptyset$, $(r_1^{(k)})_{k\in\mathbb{N}}\not = \emptyset$ and $(r_2^{(k)})_{k\in\mathbb{N}}\not = \emptyset$.
	Then, $\lim_{k\to \infty} h_k(r_0^{(k)})=0$ and using Assumption~\ref{ass: reg_level_sets} we have
	\begin{align*}
	0\leq \lim_{k\to\infty} h_k(r_1^{(k) }) 
	\leq \lim_{k\to\infty} \sup_{s\in (0,\varepsilon)} g_k(s) 
	\leq \lim_{k\to\infty} \left(1-\frac{2c}{\sigma_d}\right)^{2k} =0. 
	\end{align*}
	Moreover, by the definition of $(r_2^{(k)})_{k\in\mathbb{N}}$ 
	note that $1/r_2^{(k)}\cdot\mathbf{1}_{(0,r_2^{(k)})}(t) \leq \varepsilon^{-1}$ for $t\in(0,\infty)$,
	such that
	\begin{align*}
	\lim_{k\to \infty} h_k(r_1^{(k)}) 
	& =  \lim_{k\to \infty}\int_0^{\Vert \rho \Vert_\infty} \frac{\mathbf{1}_{(0,r_1^{(k)})}(t)}{r_1^{(k)}}   g_k(t) \dint t = \int_0^{\Vert \rho \Vert_\infty} \lim_{k\to \infty} \frac{\mathbf{1}_{(0,r_1^{(k)})} (t)}{r_1^{(k)}}  g_k(t) \dint t =0.
	\end{align*}
	Consequently $\lim_{k\to \infty} h_k(r^{(k)})=0$,
	%
	%
	such that $\lim_{k\to\infty} \beta_k \leq \lim_{k\to \infty} h_k(r^{(k)})^{1/2} =0.$
\end{proof}
This observation leads by Theorem~\ref{thm: low_upp_spec} to the following result.
\begin{corollary}
	Let $\rho \colon K \to (0,\infty)$ and let Assumption~\ref{ass: reg_level_sets} be satisfied.
	Then,
	the hit-and-run slice sampler has an absolute spectral gap
	if and only if the simple slice sampler has an absolute spectral gap. 
\end{corollary}

We stress
that we do not know whether the level sets of $\rho$ are convex, star-shaped or have any additional structure. In this sense
the imposed assumptions on $\rho$ can be considered as weak.
This also means that it might be difficult to implement hit-and-run in this generality.
In the next section we consider a combination
of hit-and-run, stepping-out and shrinkage procedure, where
we provide a concrete implementable algorithm.

\subsection{Hit-and-run, stepping-out and shrinkage slice sampler}

We combine hit-and-run, stepping-out and shrinkage procedure.
Let $w>0$, let $K\subseteq \R^d$ and assume that $\rho\colon K \to (0,\infty)$.
We say $\rho\in \mathcal{R}_{d,w}$ if the following conditions are satisfied:
\begin{enumerate}
	\item\label{en: a}
	there are not necessarily normalised lower semi-continuous and quasi-concave densities 
	$\rho_1,\rho_2 \colon K \to (0,\infty)$, i.e.
	the level sets are open and convex,
	with
	\[
	\inf_{y \in \arg \max \rho_1,\,z\in \arg \max \rho_2 } \abs{z-y} \leq \frac{w}{2}
	\]
	such that
	$\rho(x)=\max\{ \rho_1(x), \rho_2(x) \}$.
	\item\label{en: har_st_sh} 	
	the limit $\kappa:=\lim_{t\downarrow 0} \frac{\vol_d(K(t))}{\diam (K(t))^d}$ exists and there are numbers $c,\varepsilon\in (0,1]$ such that
		\[
		\inf_{t\in(0,\varepsilon)} \frac{\vol_d(K(t))}{\diam (K(t))^d} = c.
		\]
\end{enumerate}
For $i=1,2$ let the level set of $\rho_i$ be denoted by $K_i(t)$ for $t\in [0,\norm{\rho_i}_{\infty})$ and set $K_i(t) = \emptyset$ for $t\geq \norm{\rho_i}_{\infty}$.
Then, by $\rho=\max\{ \rho_1, \rho_2 \}$ follows that $K(t)=K_1(t) \cup K_2(t)$.
If $K$ is bounded and has positive Lebesgue measure, then \ref{en: har_st_sh}. is always satisfied.
For $K=\R^d$ one has to check \ref{en: har_st_sh}. 
For example $\rho\colon \R^d \to (0,\infty)$ with
\[
\rho(x) = \max\{\exp(-\a \abs{x}^2),\exp(-\beta \abs{x-x_0}^2)\} 
\]
and $2\beta> \alpha$ satisfies \ref{en: a}. and \ref{en: har_st_sh}. for $w=2\abs{x_0}$.
The rough idea
for a transition from $x$ to $y$ 
of the 
combination of the different methods on the level set $K(t)$ is as follows:
Consider a line/segment of the form
\[
L_t(x,\theta) = \{ r\in \R \mid x+r\theta \in K(t) \}.
\]
Then, run the stepping out and shrinkage procedure on $L_t(x,\theta)$ and return $y$. 
In detail, we present a single transition from $x$ to $y$ of the hit-and-run, stepping-out, shrinkage 
slice sampler in Algorithm~\ref{fig: har_st_sh}.

\begin{algorithm}
	\label{fig: har_st_sh}
	A hybrid slice sampling transition of hit-and-run, stepping-out and shrinkage procedure 
	from $x$ to $y$, i.e. input $x$ and output $y$.
	The stepping-out procedure on $L_t(x,\theta)$ (line of hit-and-run on level set) has inputs $x$,
	$w>0$ (step size parameter from $\mathcal{R}_{d,w}$) and outputs an interval $[L,R]$.
	The shrinkage procedure has input $[L,R]$ and output $y=x+s\theta$:
	\begin{enumerate}
		\item Choose a level $t \sim \mathcal{U}(0,\rho(x))$;
		\item Choose a direction $\theta \in S_{d-1}$ uniformly distributed;
		\item Stepping-out on $L_t(x,\theta)$ with $w>0$ outputs an interval $[L,R]$:
		\begin{enumerate}
			\item Choose $u \sim \mathcal{U}[0,1]$. Set $L= u w$ and $R=L+w$;
			\item Repeat until $t \geq \rho(x+L\theta)$, i.e. $L \not \in L_t(x,\theta)$:
			\quad Set $L=L-w$;
			\item Repeat until $t \geq \rho(x+R\theta)$, i.e. $R \not \in L_t(x,\theta)$:
			\quad Set $R=R+w$;
		\end{enumerate}
		\item Shrinkage procedure with input $[L,R]$ outputs $y$:
		\begin{enumerate}
			\item Set $\bar{L}=L$ and $\bar{R}=R$;
			\item Repeat:
			\begin{enumerate}
				\item   Choose $v\sim \mathcal{U}[0,1]$ and set $s=\bar{L}+ v (\bar{R}-\bar{L})$;
				\item   If $s  \in L_t(x,\theta)$ return $y=x+s\theta$ and exit the loop;
				\item   If $s<0$ then set $\bar{L}=s$, else $\bar{R}=s$.
			\end{enumerate}
		\end{enumerate}
	\end{enumerate}
\end{algorithm}

Now we present the corresponding transition kernel on $K(t)$.
Since $\rho \in \mathcal{R}_{d,w}$ we can define 
for $i=1,2$ the open intervals
\[
L_{t,i}(x,\theta) = \{ s\in \R \mid x+s\theta \in K_i(t) \} 
\]
and have 
$
L_t(x,\theta) = L_{t,1}(x,\theta) \cup L_{t,2}(x,\theta).
$
Let 
\[
\delta_{t,\theta,x} = \inf_{r \in L_{t,1}(x,\theta),\;s \in L_{t,2}(x,\theta)} \abs{r-s}.
\]
and note that if $\delta_{t,\theta,x}>0$ then $L_{t,1}(x,\theta)\cap L_{t,2}(x,\theta) = \emptyset$.

We also write for short $\abs{L_t(x,\theta)}=\vol_1(L_t(x,\theta))$ and
for $A\in \mathcal{B}(K)$, $x\in K$, $\theta \in S_{d-1}$ let $A_{x,\theta} = \{ s\in \R\mid x+s\theta \in A \}$.
With this notation, for $t>0$, the transition kernel $H_t$ on $K(t)$ is given by
\begin{align*}
& H_t(x,A) 
= \int_{S_{d-1}} \Bigg[ \gamma_t(x,\theta) \frac{\abs{L_t(x,\theta)\cap A_{x,\theta} }}{\abs{L_t(x,\theta)}} 
\\
& \qquad 	
    + (1-\gamma_t(x,\theta)) 
\sum_{i=1}^2 \mathbf{1}_{K_i(t)}(x)\frac{\abs{L_{t,i}(x,\theta)\cap A_{x,\theta}}}{\abs{L_{t,i}(x,\theta)}}
\Bigg]   \frac{\dint \theta}{\sigma_d},
\end{align*}
with
\[
\gamma_t(x,\theta) = \frac{(w-\delta_{t,x,\theta})}{w} \cdot \frac{\abs{L_{t}(x,\theta)}}{\abs{L_{t}(x,\theta)}+\delta_{t,x,\theta}}.
\]
The following result is helpful.
\begin{lemma}
	For $\rho\in \mathcal{R}_{d,w}$ and for any $t\in (0,\norm{\rho}_{\infty})$ holds:
	\begin{enumerate}
		\item \label{it: pos_har_ss}  The transition kernel $H_t$ is reversible and
		induces a positive semi-definite operator on $L_{2,t}$, 
		i.e. for $f\in L_{2,t}$ holds $\scalar{H_t f}{f}_t \geq 0$.
		\item \label{it: op_norm_har_ss} We have 
		\begin{equation} \label{eq: spec_gap_har_ss}
		\norm{H_t - U_t }_{L_{2,t} \to L_{2,t}} \leq 1-\frac{\vol_d(K(t))}{\sigma_d\; \diam(K(t))^d} ,
		\end{equation}
		in particular $\lim_{k\to \infty} \beta_k = 0$ with $\beta_k$
		defined in Theorem~\ref{thm: low_upp_spec}.
	\end{enumerate}
\end{lemma}
\begin{proof}
	First, note that $L_t(x+s\theta,\theta)=L_t(x,\theta)-s$, $\abs{L_t(x+s\theta,\theta)}=\abs{L_t(x,\theta)}$
	and $\gamma_t(x+s\theta,\theta)=\gamma_t(x,\theta)$ for any $x\in \R^d$, $\theta\in S_{d-1}$ and $s\in \R$.\\[0.25ex]
	\textbf{To \ref{it: pos_har_ss}.:}
	The reversibility of $H_t$ w.r.t. $U_t$ (in the setting of $\rho\in\mathcal{R}_{d,w}$) is inherited by the reversibility of hit-and-run and the reversibility of the combination of the stepping-out and shrinkage procedure, see Lemma~\ref{lema: 1_dim}. 
	
	We turn to the positive semi-definiteness:
	Let $C_t=\vol_d(K(t))$. We have
	\begin{align*}
	& \scalar{f}{H_t f}_{t} 
	= \int_{S_{d-1}} \int_{K(t)} \gamma_t(x,\theta)f(x)  \int_{L_t(x,\theta)} 
	f(x+r\theta) \, \frac{\dint r}{\abs{L_t(x,\theta)}} \frac{\dint x}{C_t}\, \frac{\dint \theta}{\sigma_d}  \\
	& + \sum_{i=1}^2 
	\int_{S_{d-1}} \int_{K_i(t)} (1-\gamma_t(x,\theta))f(x)  \int_{L_{t,i}(x,\theta)} 
	f(x+r\theta) \, \frac{\dint r}{\abs{L_{t,i}(x,\theta)}} \frac{\dint x}{C_t}\, \frac{\dint \theta}{\sigma_d}.
	\end{align*}
	We prove positivity of the first summand. 
	The positivity of the other two summands follows by the same arguments.
	For $\theta \in S_{d-1}$ let us define the projected set	
	\[
	P_{\theta^\bot}(K(t)) = \{ \tilde x \in \R^d \mid \tilde x \bot \theta,
	\; \exists s\in \R\;\, \mbox{s.t.} \;\, \tilde x + \theta s \in K(t) \}.
	\]
	Then
	\begin{align*}
	& \int_{S_{d-1}} \int_{K(t)} \gamma_t(x,\theta)f(x)  \int_{L_t(x,\theta)} 
	f(x+r\theta) \, \frac{\dint r}{\abs{L_t(x,\theta)}} \frac{\dint x}{C_t}\, \frac{\dint \theta}{\sigma_d}\\  
	& = \int_{S_{d-1}} \int_{P_{\theta^{\bot}}(K(t))}\int_{L_t(\tilde x,\theta)} \gamma_t(\tilde x+s\theta,\theta)f(\tilde x+s \theta)\, \times \\
	& \qquad\qquad  \int_{L_t(\tilde x+s\theta,\theta)} 
	f(\tilde x+(r+s)\theta) \, \frac{\dint r}{\abs{L_t(\tilde x+s\theta,\theta)}} \frac{\dint s\, \dint \tilde x}{C_t}\, \frac{\dint \theta}{\sigma_d}\\   
	& = \int_{S_{d-1}} \int_{P_{\theta^{\bot}}(K(t))}\int_{L_t(\tilde x,\theta)} 
	\gamma_t(\tilde x,\theta)f(\tilde x+s \theta)\, \times \\
	& \qquad\qquad  \int_{L_t(\tilde x,\theta)-s} 
	f(\tilde x+(r+s)\theta) \, \frac{\dint r}{\abs{L_t(\tilde x,\theta)-s}} \frac{\dint s\, \dint \tilde x}{C_t}\, \frac{\dint \theta}{\sigma_d}\\   
	& = \int_{S_{d-1}} \int_{P_{\theta^{\bot}}(K(t))} \frac{\gamma_t(\tilde x,\theta)}{\abs{L_t(\tilde x,\theta)}} 
	\left( \int_{L_t(\tilde x,\theta)} 
	f(\tilde x+u \theta) \dint u \right)^2
	\frac{\dint \tilde x}{C_t}\, \frac{\dint \theta}{\sigma_d} \geq 0.
	\end{align*}
	This gives that $H_t$ is positive semi-definite.\\[0.25ex]
	\textbf{To \ref{it: op_norm_har_ss}.:}
	For any $x\in K(t)$ and measurable $A\subseteq K(t)$ we have
	\begin{align*}
	& H_t(x,A)  \geq \int_{S_{d-1}} \gamma_t(x,\theta) 
	\int_{L_t(x,\theta)} \mathbf{1}_{A}(x+s\theta)\, \frac{\dint s}{\abs{L_t(x,\theta)}}\, \frac{\dint \theta}{\sigma_d}\\
	&= \int_{S_{d-1}} \int_0^\infty \gamma_t(x,\theta) \mathbf{1}_{A}(x-s\theta)\, \frac{\dint s}{\abs{L_t(x,\theta)}} \frac{\dint \theta}{\sigma_d} \\
	& \quad + \int_{S_{d-1}} \int_0^\infty \gamma_t(x,\theta) \mathbf{1}_{A}(x+s\theta)\, \frac{\dint s}{\abs{L_t(x,\theta)}} \frac{\dint \theta}{\sigma_d} \\
	& = \int_{\R^d} \frac{\gamma_t(x,\frac{y}{\abs{y}})}{\sigma_d\cdot \abs{L_t(x,\frac{y}{\abs{y}})}} 
	\frac{\mathbf{1}_{A}(x-y)}{\abs{y}^{d-1}}\, \dint y
	+ \int_{\R^d} \frac{\gamma_t(x,\frac{y}{\abs{y}})}{\sigma_d\cdot \abs{L_t(x,\frac{y}{\abs{y}})}} 
	\frac{\mathbf{1}_{A}(x+y)}{\abs{y}^{d-1}}\, \dint y\\
	& = \frac{2}{\sigma_d} \int_A \frac{\gamma_t(x,\frac{x-y}{\abs{x-y}})}{\abs{x-y}^{d-1}\abs{L_t(x,\frac{x-y}{\abs{x-y}})}}\; \dint y
	\geq \frac{\vol_d(K(t))}{\sigma_d\; \diam(K(t))^d} \cdot \frac{\vol_d(A)}{\vol_d(K(t))}. 
	\end{align*}
	Here the last inequality follows by the fact that $\delta_{t,x,\theta} \leq w/2$
	and $\abs{L_t(x,\theta)}+\delta_{t,x,\theta} \leq \diam(K(t))$.
	Thus, by \cite{MeTw09} we have uniform ergodicity and by \cite[Proposition~{3.24}]{Ru12}
	we obtain \eqref{eq: spec_gap_har_ss}. Finally, $\lim_{k\to \infty}\beta_k =0$
	follows by the same arguments as in Lemma~\ref{lem: har_conv_slice}. 
\end{proof}
This observation leads by Theorem~\ref{thm: low_upp_spec} to the following result.
\begin{corollary}
	Let $\rho \in \mathcal{R}_{d,w}$.
	Then,
	the hit-and-run, stepping-out, shrinkage slice sampler has an absolute spectral gap
	if and only if the simple slice sampler has an absolute spectral gap. 
\end{corollary}

\section{Concluding remarks}
We provide a general framework
to prove convergence results of hybrid slice sampling via spectral gap arguments.
More precisely, we state sufficient conditions for the spectral
gap of appropriately designed hybrid slice sampler to be equivalent to
the spectral gap of the simple slice sampler.
Since all Markov chains we are considering are reversible, 
this also provides a criterion for geometric ergodicity, see \cite{RoRo97}.

To illustrate how our analysis can be applied to specific hybrid
slice sampling implementations, we analyse the hit-and-run on the
slice algorithm on multidimensional targets under weak
conditions and the easily implementable stepping-out shrinkage
hit-and-run on the slice for bimodal $d$-dimensional
distributions. The latter analysis can be in principle extended to settings
with more than two modes at the price of further notational and
computational complexity. 

These examples demonstrate that robustness of the simple slice sampler
is inherited by its appropriately designed hybrid versions in realistic
computational settings and give theoretical underpinning for their use
in applications.

\appendix

\section{Technical lemmas} \label{app: tech_lemmas}

\begin{lemma}  \label{lem: tech_lemma}
	Let $H_1$ and $H_2$ be two Hilbert spaces.
	Further, let $R\colon H_2 \to H_1 $ be a bounded linear operator with adjoint $R^* \colon H_1 \to H_2$
	and let $Q\colon H_2 \to H_2$ be a bounded linear operator which is self-adjoint. 
	Then
	\[
	\norm{R Q^{k+1} R^*}_{H_1 \to H_1} \leq \norm{Q}_{H_2 \to H_2} \norm{R \abs{Q}^{k} R^*}_{H_1 \to H_1}.
	\]
	Let us additionally assume that $Q$ is positive semi-definite. Then
	\[
	\norm{R Q ^{k+1} R^*}_{H_1 \to H_1} \leq \norm{Q}_{H_2 \to H_2} \norm{R Q^{k} R^*}_{H_1 \to H_1}.
	\] 
\end{lemma}
\begin{proof}
	Let us denote the inner-products of $H_1$ by $\scalar{\cdot}{\cdot}_1$ and $H_2$ by
	$\scalar{\cdot}{\cdot}_2$. 
	By the spectral theorem for the bounded 
	and self-adjoint operator $Q\colon H_2 \to H_2$ we obtain
	\[
	\frac{ \scalar{Q R^* f}{R^* f}_{2} }{ \scalar{R^* f}{R^* f}_{2}}
	= \int_{\spec(Q)} \lambda \,\dint \nu_{Q,R^*f}(\lambda),
	\]
	where $\spec(Q)$ denotes the spectrum of $Q$ 
	and $\nu_{Q,R^*f}$ denotes
	the normalized
	spectral measure.
	Thus,
	\begin{align*}
	&\qquad   \norm{R Q ^{k+1} R^*}_{H_1 \to H_1} 
	= \sup_{\scalar{f}{f}_1 \not = 0} \frac{\abs{ \scalar{ Q^{k+1} R^* f}{ R^* f}_2} }{\scalar{f}{f}_1}\\
	& = \sup_{\scalar{f}{f}_1 \not = 0} \frac{\scalar{R^* f}{R^*f}_2}{\scalar{f}{f}_1}
	\frac{\abs{ \scalar{ Q^{k+1} R^* f}{ R^* f}_2} }{\scalar{R^* f}{R^*f}_2}\\
	& = \sup_{\scalar{f}{f}_1 \not = 0} 
	\frac{\scalar{R^* f}{R^*f}_2}{\scalar{f}{f}_1}
	\abs{\int_{\spec(Q)} \lambda^{k+1} \,\dint \nu_{Q,R^*f}(\lambda)}\\
	& \leq \norm{Q}_{H_2 \to H_2}
	\sup_{\scalar{f}{f}_1 \not = 0} 
	\frac{\scalar{R^* f}{R^*f}_2}{\scalar{f}{f}_1}
	\int_{\spec(Q)} \abs{\lambda}^{k} \,\dint \nu_{Q,R^*f}(\lambda)\\
	& = \norm{Q}_{H_2 \to H_2} \sup_{\scalar{f}{f}_1 \not = 0} 
	\frac{\scalar{R^* f}{R^*f}_2}{\scalar{f}{f}_1}
	\frac{ \scalar{\abs{ Q}^{k} R^* f}{ R^* f}_2 }{\scalar{R^* f}{R^*f}_2}\\ 
	& = \norm{Q}_{H_2 \to H_2}
	\norm{R \abs{Q} ^{k} R^*}_{H_1 \to H_1}.
	\end{align*}
	We used that the operator norm of $Q \colon H_2 \to H_2$ and the operator
	norm of $\abs{Q}\colon H_2 \to H_2$ is the same.
	If $Q$ is positive semi-definite, then $Q= \abs{Q}$.
\end{proof}

\begin{lemma} \label{lem: tech_lem_2}
	Let as assume that the conditions of Lemme~\ref{lem: tech_lemma} are satisfied. Further let
	$\norm{R}_{H_2 \to H_1}^2 = \norm{R R^* }_{H_1 \to H_1} \leq 1$.
	Then
	\[
	\norm{R Q R^*}_{H_1 \to H_1}^{k} \leq \norm{R \abs{Q} ^{{k}} R^*}_{H_1 \to H_1}.
	\]
	Let us additionally assume that $Q$ is positive semi-definite. Then
	\[
	\norm{R Q R^*}_{H_1 \to H_1}^{k} \leq \norm{R Q ^{{k}} R^*}_{H_1 \to H_1}.
	\] 
\end{lemma}
\begin{proof}
	We use the same notation as in the proof of Lemma~\ref{lem: tech_lemma}.
	Thus
	\begin{align*}
	\norm{R Q R^*}_{H_1 \to H_1} ^k 
	& = \sup_{\scalar{f}{f}_1 \not = 0} \left(\frac{\scalar{R^* f}{R^* f}_2}{\scalar{f}{f}_1} 
	\frac{\abs{  \scalar{Q R^* f}{R^* f}_{2} }}{\scalar{R^* f}{R^* f}_2}  \right)^ k\\
	& = \sup_{\scalar{f}{f}_1 \not = 0} \left(\frac{\scalar{R^* f}{R^* f}_2}{\scalar{f}{f}_1} \right)^ k
	\abs{\int_{\spec(Q)} \lambda \,\dint \nu_{Q,R^*f}(\lambda)  }^ k\\
	& \leq \sup_{\scalar{f}{f}_1 \not = 0} \left(\frac{\scalar{R^* f}{R^* f}_2}{\scalar{f}{f}_1} \right)^ k
	\int_{\spec(Q)} \abs{\lambda}^k \,\dint \nu_{Q,R^*f}(\lambda)  \\ 
	& = \sup_{\scalar{f}{f}_1 \not = 0} \left(\frac{\scalar{R^* f}{R^* f}_2}{\scalar{f}{f}_1} \right)^ k
	\frac{\scalar{\abs{Q}^k R^* f}{ R^* f}_2 }{\scalar{R^* f}{R^* f}_2}\\
	& = \sup_{\scalar{f}{f}_1 \not = 0} \left(\frac{\scalar{R^* f}{R^* f}_2}{\scalar{f}{f}_1} \right)^ {k-1}
	\frac{\scalar{\abs{Q}^k R^* f}{ R^* f}_2 }{\scalar{f}{ f}_1} \\
	& \leq \norm{R R^* }_{H_1 \to H_1} ^{k-1} \norm{R \abs{Q} ^{{k}} R^*}_{H_1 \to H_1}
	\leq \norm{R \abs{Q} ^{{k}} R^*}_{H_1 \to H_1}.
	\end{align*}
	Note that we applied Jensen inequality.
	Further, if $Q$ is positive-semidefinite then $Q= \abs{Q}$, which finishes the proof.
\end{proof}


\acks \noindent We are extremely grateful for the careful reading of the referees and their comments. In particular, we thank one of the referees who pointed us to Example~\ref{ex: Ass2}. KL was supported by the Royal Society through the Royal Society University Research Fellowship. DR gratefully acknowledges support of the DFG within project 522337282.


\begin{thebibliography}{10}
	
	\bibitem{Ba05}
	{\sc Baxendale, P.} (2005).
	\newblock Renewal theory and computable convergence rates for geometrically
	ergodic {M}arkov chains.
	\newblock {\em Ann. Appl. Probab.\/} {\bf 15,} 700--738.
	
	\bibitem{bednorz2007few}
	{\sc Bednorz, W. and {\L}atuszy\'{n}ski, K.} (2007).
	\newblock {A few remarks on Fixed-width output analysis for Markov chain Monte
		Carlo by Jones et al}.
	\newblock {\em {J. Amer. Statist. Assoc.}\/} {\bf 102,} 1485--1486.
	
	\bibitem{BeRoSm93}
	{\sc B{\'e}lisle, C., Romeijn, H. and Smith, R.} (1993).
	\newblock Hit-and-run algorithms for generating multivariate distributions.
	\newblock {\em Math. Oper. Res.\/} {\bf 18,} 255--266.
	
	\bibitem{besag1993spatial}
	{\sc Besag, J. and Green, P.} (1993).
	\newblock Spatial statistics and {B}ayesian computation.
	\newblock {\em J. Roy. Statist. Soc. Ser. B\/} 25--37.
	
	\bibitem{DiFr97}
	{\sc Diaconis, P. and Freedman, D.} (1997).
	\newblock On the hit and run process.
	\newblock {\em Tech. rep. 497, Dep. Stat., University of California,
		Berkeley\/} 1--16.
	
	\bibitem{edwards1988generalization}
	{\sc Edwards, R. and Sokal, A.} (1988).
	\newblock {Generalization of the Fortuin-Kasteleyn-Swendsen-Wang representation
		and Monte Carlo algorithm}.
	\newblock {\em Phys. Rev. D\/} {\bf 38,} 2009.
	
	\bibitem{flegal2010batch}
	{\sc Flegal, J. and Jones, G.} (2010).
	\newblock {Batch means and spectral variance estimators in Markov chain Monte
		Carlo}.
	\newblock {\em Ann. Statist.\/} {\bf 38,} 1034--1070.
	
	\bibitem{Ge92}
	{\sc Geyer, C.} (1992).
	\newblock Practical {M}arkov chain {M}onte {C}arlo.
	\newblock {\em Statist. Sci.\/} {\bf 7,} 473--483.
	
	\bibitem{Hi98}
	{\sc Higdon, D.} (1998).
	\newblock Auxiliary variable methods for {Markov chain Monte Carlo} with
	applications.
	\newblock {\em J. Amer. Statist. Assoc.\/} {\bf 93,} 585--595.
	
	\bibitem{hobert2002applicability}
	{\sc Hobert, J., Jones, G., Presnell, B. and Rosenthal, J.} (2002).
	\newblock {On the applicability of regenerative simulation in Markov chain
		Monte Carlo}.
	\newblock {\em Biometrika\/} {\bf 89,} 731.
	
	\bibitem{JoHaCaNe06}
	{\sc Jones, G., Haran, M., Caffo, B. and Neath, R.} (2006).
	\newblock Fixed-width output analysis for {M}arkov chain {M}onte {C}arlo.
	\newblock {\em J. Amer. Statist. Assoc.\/} {\bf 101,} 1537--1547.
	
	\bibitem{KaSm98}
	{\sc Kaufman, D. and Smith, R.} (1998).
	\newblock Direction choice for accelerated convergence in hit-and-run sampling.
	\newblock {\em Oper. Res.\/} {\bf 46,} 84--95.
	
	\bibitem{KiSmZa11}
	{\sc Kiatsupaibul, S., Smith, R. and Zabinsky, Z.} (2011).
	\newblock An analysis of a variation of hit-and-run for uniform sampling from
	general regions.
	\newblock {\em ACM Trans. Model. Comput. Simul.\/} {\bf 21,} 16:1--16:11.
	
	\bibitem{kipnis1986central}
	{\sc Kipnis, C. and S, V.} (1986).
	\newblock {Central limit theorem for additive functionals of reversible Markov
		processes and applications to simple exclusions}.
	\newblock {\em Comm. Math. Phys.\/} {\bf 104,} 1--19.
	
	\bibitem{kontoyiannis2009geometric}
	{\sc Kontoyiannis, I. and Meyn, S.} (2012).
	\newblock {Geometric ergodicity and the spectral gap of non-reversible Markov
		chains}.
	\newblock {\em Probab. Theory Related Fields\/} {\bf 154,} 327--339.
	
	\bibitem{Lo99}
	{\sc Lov{\'a}sz, L.} (1999).
	\newblock Hit-and-run mixes fast.
	\newblock {\em Math. Program.\/} {\bf 86,} 443--461.
	
	\bibitem{LoVe06}
	{\sc Lov{\'a}sz, L. and Vempala, S.} (2006).
	\newblock Hit-and-run from a corner.
	\newblock {\em SIAM J. Comput.\/} {\bf 35,} 985--1005.
	
	\bibitem{MeTw09}
	{\sc Meyn, S. and Tweedie, R.} (2009).
	\newblock {\em {M}arkov {c}hains and {s}tochastic {s}tability} second~ed.
	\newblock Cambridge University Press.
	
	\bibitem{Mi01}
	{\sc Mira, A.} (2001).
	\newblock {Ordering and improving the performance of Monte Carlo Markov
		chains}.
	\newblock {\em Statist. Sci.\/} {\bf 16,} 340--350.
	
	\bibitem{mira2001perfect}
	{\sc Mira, A., M{\o}ller, J. and Roberts, G.} (2001).
	\newblock Perfect slice samplers.
	\newblock {\em J. Roy. Statist. Soc. Ser. B\/} {\bf 63,} 593--606.
	
	\bibitem{MiTi02}
	{\sc Mira, A. and Tierney, L.} (2002).
	\newblock Efficiency and convergence properties of slice samplers.
	\newblock {\em Scand. J. Statist.\/} {\bf 29,} 1--12.
	
	\bibitem{MuAdMa10}
	{\sc Murray, I., Adams, R. and MacKay, D.} (2010).
	\newblock Elliptical slice sampling.
	\newblock In {\em The Proceedings of the 13th International Conference on
		Artificial Intelligence and Statistics}.
	\newblock vol.~9 of {\em JMLR: W\&CP}.
	\newblock pp.~541--548.
	
	\bibitem{natarovskii2021geometric}
	{\sc Natarovskii, V., Rudolf, D. and Sprungk, B.} (2021).
	\newblock Geometric convergence of elliptical slice sampling.
	\newblock In {\em Proceedings of the 38th International Conference on Machine
		Learning}.
	\newblock vol.~139 of {\em PMLR}.
	\newblock pp.~7969--7978.
	
	\bibitem{natarovskii2021quantitative}
	{\sc Natarovskii, V., Rudolf, D. and Sprungk, B.} (2021).
	\newblock {Quantitative spectral gap estimate and Wasserstein contraction of
		simple slice sampling}.
	\newblock {\em Ann. Appl. Probab.\/} {\bf 31,} 806 -- 825.
	
	\bibitem{Ne93}
	{\sc Neal, R.} (1993).
	\newblock {\em {Probabilistic inference using Markov chain Monte Carlo
			methods}}.
	\newblock Department of Computer Science, University of Toronto, ON, Canada.
	
	\bibitem{Ne03}
	{\sc Neal, R.} (2003).
	\newblock Slice sampling.
	\newblock {\em Ann. Statist.\/} {\bf 31,} 705--767.
	\newblock With discussions and a rejoinder by the author.
	
	\bibitem{NoRu14}
	{\sc Novak, E. and Rudolf, D.} (2014).
	\newblock {Computation of expectations by Markov chain Monte Carlo methods}.
	\newblock {\em Extraction of Quantifiable Information from Complex Systems\/}
	397--411.
	
	\bibitem{RoRo97}
	{\sc Roberts, G. and Rosenthal, J.} (1997).
	\newblock Geometric ergodicity and hybrid {M}arkov chains.
	\newblock {\em Electron. Comm. Probab.\/} {\bf 2,} no.\ 2, 13--25.
	
	\bibitem{RoRo99}
	{\sc Roberts, G. and Rosenthal, J.} (1999).
	\newblock Convergence of slice sampler {M}arkov chains.
	\newblock {\em J. R. Stat. Soc. Ser. B Stat. Methodol.\/} {\bf 61,} 643--660.
	
	\bibitem{RoRo02}
	{\sc Roberts, G. and Rosenthal, J.} (2002).
	\newblock The polar slice sampler.
	\newblock {\em Stoch. Models\/} {\bf 18,} 257--280.
	
	\bibitem{RoRo08}
	{\sc Roberts, G. and Rosenthal, J.} (2008).
	\newblock Variance bounding {M}arkov chains.
	\newblock {\em Ann. Appl. Probab.\/} {\bf 18,} 1201--1214.
	
	\bibitem{Ru91}
	{\sc Rudin, W.} (1991).
	\newblock {\em Functional analysis} second~ed.
	\newblock International Series in Pure and Applied Mathematics. McGraw-Hill
	Inc., New York.
	
	\bibitem{Ru12}
	{\sc Rudolf, D.} (2012).
	\newblock Explicit error bounds for {M}arkov chain {M}onte {C}arlo.
	\newblock {\em Dissertationes Math.\/} {\bf 485,} 93 pp.
	
	\bibitem{Ru13}
	{\sc Rudolf, D.} (2013).
	\newblock Hit-and-run for numerical integration.
	\newblock {\em J. Dick, F. Kuo, G. Peters, and I. Sloan (Eds.) - Monte Carlo
		and Quasi-Monte Carlo Methods 2012. Springer, Berlin\/} 565--579.
	
	\bibitem{RuUl13}
	{\sc Rudolf, D. and Ullrich, M.} (2013).
	\newblock {Positivity of hit-and-run and related algorithms}.
	\newblock {\em Electron. Commun. Probab.\/} {\bf 18,} 1--8.
	
	\bibitem{RuUl15}
	{\sc Rudolf, D. and Ullrich, M.} (2018).
	\newblock {Comparison of hit-and-run, slice sampler and random walk
		Metropolis}.
	\newblock {\em J. Appl. Probab.\/} {\bf 55,} 1186--1202.
	
	\bibitem{pmlr-v202-schar23a}
	{\sc Sch\"{a}r, P., Habeck, M. and Rudolf, D.} (2023).
	\newblock Gibbsian polar slice sampling.
	\newblock In {\em Proceedings of the 40th International Conference on Machine
		Learning}.
	\newblock vol.~202.
	\newblock PMLR.
	\newblock pp.~30204--30223.
	
	\bibitem{Sm84}
	{\sc Smith, R.} (1984).
	\newblock Efficient {M}onte {C}arlo procedures for generating points uniformly
	distributed over bounded regions.
	\newblock {\em Oper. Res.\/} {\bf 32,} 1296--1308.
	
	\bibitem{Ul12}
	{\sc {Ullrich}, M.} (2012).
	\newblock {Rapid mixing of Swendsen-Wang and single-bond dynamics in two
		dimensions}.
	\newblock {\em ArXiv e-prints\/}.
	
	\bibitem{Ul14}
	{\sc {Ullrich}, M.} (2014).
	\newblock {Rapid mixing of Swendsen-Wang dynamics in two dimensions}.
	\newblock {\em Dissertationes Math.\/} {\bf 502,} 64 pp.
	
	\bibitem{wang1987nonuniversal}
	{\sc Wang, J. and Swendsen, R.} (1987).
	\newblock {Nonuniversal critical dynamics in Monte Carlo simulations}.
	\newblock {\em Phys. Rev. Lett.\/}.
	
	\bibitem{We05}
	{\sc Werner, D.} (2005).
	\newblock {\em Funktionalanalysis}.
	\newblock Springer Verlag.
	
\end{thebibliography}
\end{document}